\documentclass{article}

\usepackage{arxiv}
\usepackage{amsmath}
\usepackage{diagbox}
\usepackage{comment}
\usepackage{amsmath}
\usepackage{amsfonts}
\usepackage{amssymb}
\usepackage{multirow}
\usepackage{multicol}

\usepackage[utf8]{inputenc} % allow utf-8 input
\usepackage[T1]{fontenc}    % use 8-bit T1 fonts
\usepackage{hyperref}       % hyperlinks
\usepackage{url}            % simple URL typesetting
\usepackage{booktabs}       % professional-quality tables
\usepackage{nicefrac}       % compact symbols for 1/2, etc.
\usepackage{microtype}      % microtypography
\usepackage{lipsum}
\usepackage{graphicx}
\graphicspath{{media/}}     % organize your images and other figures under media/ folder
\usepackage[table]{xcolor}
\definecolor{gw}{rgb}{0.97, 0.97, 1.0}
\usepackage{cancel}
\usepackage{enumitem}
\usepackage{arydshln}
\usepackage{amsthm}
\theoremstyle{plain}
\newtheorem{theorem}{Theorem}[section]
\newtheorem{corollary}[theorem]{Corollary}
\newtheorem{remark}[theorem]{Remark}

\usepackage{textcomp, gensymb}
\usepackage{nicematrix}

\usepackage{appendix}

\title{Independence Testing for Mixed Data}

\author{ 
\href{https://orcid.org/0000-0002-1999-969X}{\includegraphics[scale=0.06]{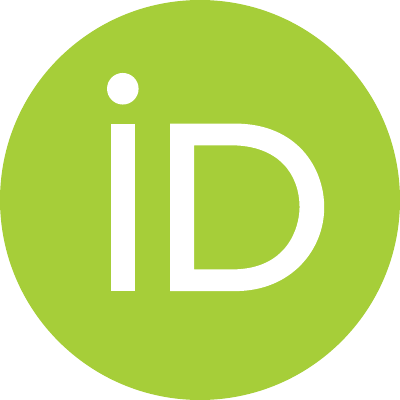}\hspace{1mm}Dana Bucalo Jeli\'c} \\
	University of Belgrade\\
    PhD student at the Faculty of Mathematics\\
	Belgrade, 11000, Serbia \\
	\texttt{bucalo@agrif.bg.ac.rs} \\
 \And
 \href{https://orcid.org/0000-0001-5071-8350}{\includegraphics[scale=0.06]{orcid.pdf}\hspace{1mm} Marija Cupari\' c} \\
	University of Belgrade\\
	Faculty of Mathematics\\
	Belgrade, 11000, Serbia \\
	\texttt{marija.cuparic@matf.bg.ac.rs} \\
	\And
    	\href{https://orcid.org/0000-0001-8243-9794}{\includegraphics[scale=0.06]{orcid.pdf}\hspace{1mm}Bojana Milo\v sevi\' c} \\
	University of Belgrade\\
        Faculty of Mathematics\\
	Belgrade, 11000, Serbia \\
	\texttt{bojana@matf.bg.ac.rs} \\
}

\date{}

\begin{document}
\maketitle

\begin{abstract}
We consider the problem of testing independence in mixed-type data that combine count variables with positive, absolutely continuous variables. We first introduce two distinct classes of test statistics in the bivariate setting, designed to test independence between the components of a bivariate mixed-type vector. These statistics are then extended to the multivariate context to accommodate: (i) testing independence between vectors of different types and possibly different dimensions, and (ii) testing total independence among all components of vectors with different types. The construction is based on the recently introduced Baringhaus–Gaigall transformation, which characterizes the joint distribution of such data. We establish the asymptotic properties of the resulting tests and, through an extensive power study, demonstrate that the proposed approach is both competitive and flexible.
\end{abstract}

% keywords can be removed

\section{Introduction}

Independence plays a central role in statistics and many related fields, underpinning much of statistical theory.
Since then, numerous new measures of dependence have been developed and studied, each with its own advantages and disadvantages, and no single measure has been universally accepted. For an overview we refer to \cite{samuel2001correlation} and \cite{josse2016measuring}. Alongside these developments, many statistical tests for assessing independence have also been proposed. 
Independence between categorical outcomes is often tested by classical tests such as Pearson’s $\chi^2$ or the likelihood ratio statistic (see \cite{agresti2013categorical} and \cite{berrett2021usp}). There are many tests based on empirical distribution function (see e.g. \cite{hoeffding1994non,blum1961distribution,de1980cramer,genest2004test}), those utilizing empirical copula (see e.g. \cite{GNR19}) or integral transforms (see e.g. \cite{csorgHo1985testing, GFL19,bucalo2024testing}), mostly designed for ordinal categorical data,  continuous or discrete data. Since it was shown to be equivalent to the $L^2$-weighted difference between empirical characteristic functions, the notion of distance covariance and its generalizations has become very popular for testing independence (see \cite{SRB07,sejdinovic2013equivalence,edelmann2021relationships,cuparic2024new,edelmann2025generalized}). The main reason is its flexibility with respect to the types of data to which the statistic can be applied. All those tests were proposed against a general alternative of non-independence.
Some of the proposed tests were specifically designed to test against special notions of stochastic dependence (\cite{kochar1990distribution, shetty2003distribution} 
against positive quadrant dependence, \cite{amini2020new} against positive regression dependence etc.). We address two related problems: testing independence between mixed-type vectors and testing total independence. For methodological clarity, we begin with the bivariate setting and subsequently present the corresponding multivariate extensions.

Let $(X, Y)$ be a bivariate random vector such that $X > 0$ and $Y$ takes values in $\mathbb{N}_0 = \mathbb{N} \cup \{0\}$, where $\mathbb{N}$ denotes the set of natural numbers. We assume (and throughout the paper will maintain) that $X = 0$ if and only if $Y = 0$. Consequently, $(X, Y)$ takes values in the Borel subset
$
R = ((0, \infty) \times \mathbb{N}) \cup \{(0, 0)\}
$
of $\mathbb{R}^2$. We further assume that $\mathbb{E}(X) < \infty$ and $\mathbb{E}(Y) < \infty$.

Let $\mathbb{R}_{+}$ be the set of non-negative real numbers. The function $\psi: \mathbb{R}_{+} \times [0,1] \to \mathbb{R}_{+}$ is defined by (see \cite{LB23})
$$ \psi(s,t) = E(e^{-sX} t^Y), \; s \geq 0, \; 0 \leq t \leq 1. $$

Let $L(s)$ denote the Laplace transform of a random variable $X$, defined by 
$ L(s) = E(e^{-sX})$, and let $G(t)$ denote the probability generating function of a count random variable $Y$, defined by $G(t) = E(t^Y)$. Notice that $L(s)=\psi(s,1)$ and
$G(t)=\psi(0,t)$. Random variable $X$ and $Y$ are independent if and only if, for each $ s \geq 0, \; 0 \leq t \leq 1$ we have 
$$\psi(s,t)=%E(e^{-sX} t^Y) = E(e^{-sX}) \cdot E(t^Y)=
L(s) G(t). $$
Let $(X_1,Y_1),...,(X_n,Y_n)$ be independent identically distributed (i.i.d.) bivariate sample equally distributed as $(X,Y)$. Denote the empirical counterparts of previous functions with
\begin{align*}
\psi_n(s,t)&=\frac{1}{n} \sum\limits_{i=1}^{n} e^{-sX_{i}}t^{Y_i},\\
L_n(s) &= \frac{1}{n}\sum\limits_{i=1}^{n} e^{-sX_{i}},\\ 
G_n(t) &= \frac{1}{n}\sum\limits_{i=1}^{n} t^{Y_{i}}.
\end{align*}
Considering this, for testing $H_0:$ that $X$ and $Y$ are independent, against that is not the case, we propose a class of test statistics
\begin{align}\label{testT}
    T_n =& \int_0^{\infty}\int_0^1(\psi_n(s,t) - L_n(s) G_n(t))^2w(s,t)dsdt,
\end{align}
where $w:\mathbb{R}_+\times [0,1] \to\mathbb{R}_{+}$  is a weight function such that $\int_0^{\infty}\int_0^1w(s,t)dsdt<\infty$. Following \cite{LB23} we choose $w(s,t)=t^b e^{-a s}$, where $a,b >0$. Then the test statistic $T_n$ has a simple form
\begin{align*}
    T_n
   =& \frac{1}{n^2} \sum_{i=1}^{n} \sum_{j=1}^{n} \frac{1}{(X_{i} + X_{j}+a)(Y_{i} + Y_{j} +b  + 1)}\\
    &+\frac{1}{n^{4}}\sum_{i=1}^{n} \sum_{j=1}^{n} \frac{1}{(X_{i} + X_{j}+a)}\sum_{i=1}^{n} \sum_{j=1}^{n}\frac{1}{(Y_{i} + Y_{j} +b  + 1)}\\
    &-2\frac{1}{n^{3}} \sum_{i=1}^{n} \sum_{j=1}^{n}\sum_{l=1}^{n}
 \frac{1}{(X_{i} + X_{j}+a)(Y_{i} + Y_{l}  
 +b + 1)}.
 \end{align*}
 It is clear that, under $H_0$, the values of $T_n$ will be small. Therefore, the null hypothesis will be rejected for large values of $T_n$.
 
One can also consider
 \begin{align}\label{testIn}
 I_{n}=&\int_0^{\infty}\int_0^1(\psi_{n}(s, t) - L_n(s) G_n(t))w(s,t)dsdt.
\end{align}
After calculation, for the previous choice of weight function, we have
 \begin{align*} 
I_n&  =
  \frac{1}{n}\sum_{i=1}^{n} \frac{1}{(X_i + a)(Y_i + b  + 1)}- \frac{1}{n} \sum_{i=1}^{n}\frac{1}{(X_i +a)}\frac{1}{n}\sum_{j=1}^{n}\frac{1}{(Y_j +b  + 1)}.
\end{align*}

Notice that $I_n$ can be close to 0 also in cases where the null hypothesis does not hold. However, it usually does not occur for the majority of alternatives, and hence we may assume large values of $\mid I_n\mid$ to be significant. Moreover, its simple form and nice properties, which we will show later, make it worth considering. 

In the next section, we derive the asymptotic null distribution of the introduced statistics. These results may be helpful for approximating the $p$-values. Then, in Section~\ref{sec: general}, we go beyond the bivariate case by considering the problem of testing the independence of two vectors of arbitrary dimensions, as well as the problem of testing total independence. Accompanying limiting results are also presented. Section~\ref{sec:power} contains the results of a comprehensive power study, in which the competitiveness of the tests is demonstrated across various settings. A recommendation for the selection of tuning parameters is also provided. The applicability of the proposed statistics is further demonstrated in Section~\ref{sec: real}, where they are applied to real data sets.

\section{Asymptotic properties of $T_n$ and $I_n$}

Let us consider Hilbert space $\mathbb{H}_w=L^2(\mathbb{R}^+\times[0,1])$, where scalar product $\langle,\rangle_w$ is defined by $\langle f,g\rangle_w=\int_0^\infty\int_0^1 f(s,t)g(s,t)w(s,t)dsdt$, and $\|\cdot\|_w$ is the corresponding $L^2$ norm. Then $T_n$ and $I_n$ can be represented as
\begin{align*}
    T_n&=\| \xi_n\|^2_w,\\
    I_n&=\langle \xi_n,1\rangle_w,
\end{align*}
where $\xi_n\in \mathbb{H}_w$ is defined with  
\begin{align*}
    \xi_n(s,t)=\frac{1}{n} \sum\limits_{i=1}^{n} e^{-sX_{i}}t^{Y_i}-\frac{1}{n^2}\sum\limits_{i=1}^{n} e^{-sX_{i}}\sum\limits_{i=1}^{n}t^{Y_i}.
\end{align*}

Now we have all ingredients to formulate results about almost sure convergence and null distributions of $T_n$ and $I_n$.  
\begin{theorem}\label{as}
    Let $X_1,X_2,...,X_n$ be an i.i.d. sample from a nonnegative continuous random variable $X$, and let $Y_1,Y_2,...,Y_n$ be i.i.d. sample from nonnegative count random variable $Y$, both having finite second moments. Then
    \begin{itemize}
        \item $T_n \overset{a.s}{\to}\int_0^{\infty}\int_0^1 (\psi(s,t)-L(s)G(t))^2
 w(s,t)dtds=\|\psi-LG\|^2_w$,
        \item $I_n \overset{a.s}{\to}\int_0^{\infty}\int_0^1 (\psi(s,t)-L(s)G(t))
 w(s,t)dtds=\langle \psi-LG, 1 \rangle_w$.
    \end{itemize}
\end{theorem}

\begin{proof}[\bf  Proof of Theorem \ref{as}]
Since $\xi_n(s,t)$ is an element of Hilbert space $\mathbb{H}_w$, by the strong law of large numbers in the Hilbert space $\mathbb{H}_w$, $\| \xi_n-(\psi-LG)\|_w\overset{a.s.}{\to}0,n\to\infty$, and therefore $T_n=\| \xi_n\|_w^2\overset{a.s.}{\to}\|\psi-LG\|^2_w$. On the other hand, by applying Cauchy–Schwarz inequality, $\mid I_n-\langle \psi-LG,1\rangle_w\mid^2\leq \|\xi_n-(\psi-LG)\|^2_w$. Therefore, $I_n\overset{a.s.}{\to}\langle \psi-LG, 1 \rangle_w$.
\end{proof}

\begin{theorem}\label{proc}
    Let the null hypothesis hold. Then there exists a centered Gaussian process $\xi\in\mathbb{H}_w$ with covariance function 
    \begin{align}\label{cov}
        K((s_1,t_1),(s_2,t_2)) %=& L(s_1+s_2)G(t_1t_2)+L(s_1)L(s_2)G(t_1)G(t_2)
                            %-L(s_1+s_2)G(t_1)G(t_2)-L(s_1)L(s_2)G(t_1t_2)\\
                            =&\big(G(t_1t_2)-G(t_1)G(t_2)\big)\big(L(s_1+s_2)-L(s_1)L(s_2)\big)
    \end{align}
    such that $\sqrt{n}\xi_n$ weakly converges to $\xi$. 
\end{theorem}
\begin{proof}[\bf  Proof of Theorem \ref{proc}]
\normalfont
It is clear that $\xi_n(s,t)\in\mathbb{H}_\omega$. For the rest of the proof, we note that for each $(s,t)\in([0,\infty)\times[0,1])$ $\xi_n(s,t)$ is $V-$statistic of order 2 with kernel
\begin{align*}
    \Phi((X_1,Y_1),(X_2,Y_2);s,t)& 
 = \frac{1}{2}\big( e^{-sX_1}t^{Y_1}+e^{-sX_2}t^{Y_2}-e^{-sX_1}t^{Y_2}-e^{-sX_2}t^{Y_1}\big).
\end{align*}
Under $H_0$ we have that $E(\Phi((X_1,Y_1),(X_2,Y_2);s,t))=0$. Furthermore, its first projection is equal to
\begin{align*}
    \varphi_1(x,y;s,t)&=E(\Phi((X_1,Y_1),(X_2,Y_2);s,t)\vert X_{1}=x,Y_{1}=y)\\
  %&=\frac{1}{2}E\big( e^{-sx}t^{y}+e^{-sX_2}t^{Y_2}-e^{-sx}t^{Y_2}-e^{-sX_2}t^{y}\big)\\ 
    &=\frac{1}{2}\big(e^{-sx}t^{y}+\psi(s,t)-e^{-sx}G(t)-L(s)t^{y}\big).
\end{align*}
Its variance is equal to
\begin{align*}
  E(\varphi_1  (X,Y;s,t))^2&=\frac{1}{4}E\bigg(e^{-sX}t^{Y}+\psi(s,t)-e^{-sX}G(t)-L(s)t^{Y}\bigg)^2\\
%&=\frac{1}{4}E\bigg( (e^{-sx}t^{y})^2+(\psi(s,t))^2+(e^{-sx}G(1,t))^2+(L(s,0)t^{y})^2\\
% &+2\big((e^{-sx}t^{y})\psi(s,t)-e^{-sx}t^{y}e^{-sx}G(1,t)-e^{-sx}t^{y}L(s,0)t^{y}\\
% &-\psi(s,t)e^{-sx}G(1,t)-\psi(s,t)L(s,0)t^{y}+e^{-sx}G(1,t)L(s,0)t^{y}\big)\bigg)\\
% &=\frac{1}{4}\bigg(\psi(2s,t^2)+{\color{red}{(\psi(s,t))^2}}+L(2s,0)G^2(1,t)+L^2(s,0)G(1,t^2)\\
% &+2\big({\color{red}{(\psi(s,t))^2}}-{\color{blue}{L(2s,0)G^2(1,t)}}-{\color{gray}{L^2(s,0)G(1,t^2)}}\\
% &-{\color{red}{\psi(s,t)L(s,0)G(1,t)}}-{\color{red}{\psi(s,t)L(s,0)G(1,t)}} +{\color{red}{L^2(s,0)G^2(1,t)}}\big)\bigg)\\
% &=\frac{1}{4}\bigg(\psi(2s,t^2)+(\psi(s,t))^2-L(2s,0)G^2(1,t)-L^2(s,0)G(1,t^2) \bigg)\\
 &=\frac{1}{4}\bigg(L(2s)G(t^2)+L^2(s)G^2(t)-L(2s)G^2(t)-L^2(s)G(t^2)
 \bigg)\\
 &=\frac{1}{4}\bigg(G(t^2)-G^2(t)\bigg)\bigg(L(2s)-L^2(s)\bigg),
      \end{align*}
which is finite. Therefore, $\xi_n(s,t)$ is nondegenerate $V-$statistic with finite variance, and consequently, $\sqrt{n}\xi_n$ has a limiting normal $\mathcal{N}(0,(G(t^2)-G^2(t))(L(2s)-L^2(s)))$ distribution. Moreover, $\sqrt{n}\xi_n(s,t)$ can be represented as
\begin{align*}
    \sqrt{n}\xi_n(s,t)&=\frac{2}{\sqrt{n}}\sum\limits_{i=1}^n\varphi_1(X_i,Y_i;s,t)+r_n(s,t),
\end{align*}
such that $\|r_n(s,t)\|_w=o_p(1)$, from which we obtain that 
\begin{align*}
    cov(\sqrt{n}\xi_n(s_1,t_1),\sqrt{n}\xi_n(s_2,t_2))\to K((s_1,t_1),(s_2,t_2))&=4E(\varphi_1(X_1,Y_1;s_1,t_1)\varphi_1(X_1,Y_1;s_2,t_2))\\&=\big(G(t_1t_2)-G(t_1)G(t_2)\big)\big(L(s_1+s_2)-L(s_1)L(s_2)\big)
\end{align*}

Since
\begin{align*}
    \int_0^\infty\int_0^1 K((s,t),(s,t))w(s,t)dtds&=\int_0^\infty\int_0^1\big(G(t^2)-(G(t))^2\big)\big(L(2s)-(L(s))^2\big)w(s,t)dtds<\infty,
\end{align*}
the proof follows.
\end{proof}

As a consequence, we have the following statement.
\begin{corollary}\label{posledica}
    Let the null hypothesis hold. Then 
    \begin{itemize}
        \item $nT_n\overset{d}{\to}\sum\limits_{n=1}^\infty \lambda_{i,\omega}Z_i^2$, where $\lambda_{1,\omega}\geq \lambda_{2,\omega}\geq ...$ is the sequence of eigenvalues of covariance operator defined with kernel \eqref{cov} and weight function $w$, and $Z_1, Z_2,...$ is the sequence of independent standard Gaussian random variables;
        \item $\sqrt{n}\frac{I_n}{\sigma_w}\overset{d}{\to}Z$, where $Z$ is standard Gaussian random variable and 
\begin{align}\label{disperzija}
        \sigma^2_w=\int_0^\infty\int_0^1\int_0^\infty\int_0^1 K((s_1,t_1),(s_2,t_2))w(s_1,t_1)w(s_2,t_2)dt_1ds_1dt_2ds_2.
\end{align}
        Specifically, for 
        $w(s,t)=t^b e^{-a s}$
        
\begin{align}\label{disperzija2}
        \sigma^2_{a,b}=Var\Big(\frac{1}{X + a}\Big)Var\Big(\frac{1}{Y +b  + 1}\Big).
       % \Bigg(E\frac{1}{(X + a)^2}- \bigg(E\frac{1}{X + a}\bigg)^2\Bigg)
  % \Bigg(E\bigg(\frac{1}{Y +b  + 1}\bigg)^2-\bigg(E\frac{1}{Y +b  + 1}\bigg)^2\Bigg)
\end{align}
    \end{itemize}
\end{corollary}
\begin{proof}[\bf  Proof of Corollary \ref{posledica}]
\begin{itemize}
    \item Since $nT_n=\|\sqrt{n}\xi_n\|_w^2$, $\xi_n \in \mathbb{H}_\omega$ and function $x \to \| x \|_w^2$ for $x\in \mathbb{H}_\omega$ is continuous, $nT_n\overset{d}{\to}\| \xi \|_w^2$. Given that covariance function is bounded, symmetric and positive definite combining Merecer's theorem and Krahunen-Loeve expansion process 
\begin{align}\label{xi_razvoj}
\xi(s,t)=\sum_{i=1}^\infty \sqrt{\lambda_{i,w}} Z_i q_{i,w}(s,t),    
\end{align}
where $\{\lambda_{i,w}\}$ is  the sequence of eigenvalues of integral operator with kernel $K$ defined in \eqref{cov} and $\{q_{i,w}\}$ is the sequence corresponding  eigenfunctions, and 
$\{Z_i\},$ is the sequence of independent standard Gaussian random variables. Substituting the expression \eqref{xi_razvoj} in $\|\xi\|^2_w$ we obtain the limiting distribution of $nT_n.$
    \item Since $I_n=\langle\xi_n,1\rangle$, $\sqrt{n}I_n\overset{d}{\to}\mathcal{N}(0,\sigma^2_w)$ where $\sigma^2_w$ is given in \eqref{disperzija}. 
Having in mind the definitions of $L(\cdot)$ and $G(\cdot)$, an interchange of expectations give us \eqref{disperzija2}. 
\end{itemize}
   
\end{proof}

\begin{remark}
    Computing the eigenvalues $\{\lambda_{i,\omega},i=1,2,...\}$ is a challenging task. For more details, we refer to \cite{ebner2024eigenvalues} and the references therein.
\end{remark}
\begin{remark}
    Although $\sigma^2_{a,b}$ from \eqref{disperzija2} is unknown it can be estimated with 
    \begin{align}\label{sigmaab}
         \widehat{\sigma}^2_{a,b}=\frac{1}{(n-1)^2}\sum\limits_{i=1}^n\Big(\frac{1}{X_i + a}-\frac{1}{n}\sum\limits_{j=1}^n\frac{1}{X_j + a}\Big)^2\sum\limits_{i=1}^n\Big(\frac{1}{Y_i +b  + 1}-\frac{1}{n}\sum\limits_{j=1}^n\frac{1}{Y_j +b  + 1}\Big)^2.
    \end{align}
    Moreover, following the approach from \cite{bucalo2024testing}, using auxiliary results from \cite{BN14}, we can show that $E(\widehat{\sigma}_{a,b}^2)={\sigma}_{a,b}^2(1+o(1))$ and $Var(\widehat{\sigma}_{a,b}^2)=o({\sigma}_{a,b}^4)$, and hence the $\widehat{\sigma}_{a,b}$  is a ratio-consistent estimator of ${\sigma}_{a,b}$, i.e. ${\widehat{\sigma}_{a,b}}/{\sigma_{a,b}}\xrightarrow{P}1$.
\end{remark}

Figure \ref{sampleNIn} illustrates the convergence of the null distribution of $\sqrt{n}\frac{I_n}{\widehat{\sigma}_{a,b}}$, labeled with $st.I_n$, when $X$ has exponential $E(1.5)$ distribution and $Y$ has Poisson ${P}(2)$ distribution. Table \ref{tab:q0.99} provides deeper insight into the convergence of the 0.95 and 0.99 quantiles of $\sqrt{n}\frac{\mid I_n\mid}{\widehat{\sigma}_{a,b}}$ for various marginal distributions of $X$ and $Y$ under the null hypothesis of independence. The finite sample distributions of test statistics are estimated using the Monte Carlo approach with $N=10000$ replicates. 

%\begin{figure}[!htb]
%\centering
%{%
%\resizebox*{10cm}{!}
%{\includegraphics[trim=1.1cm 1.5cm 1.2cm 1.9cm, clip=true]{T_n^2.eps}}}
%\caption{$T_n,$  $(a, b)=(10,0.5)$, for $Exp(2)P(1)$, $N=1000$} 
%\label{sampleN}
%\end{figure}

\begin{figure}[!htb]
\centering
{%
%\resizebox*{12cm}{!}
{\includegraphics[width=8cm]{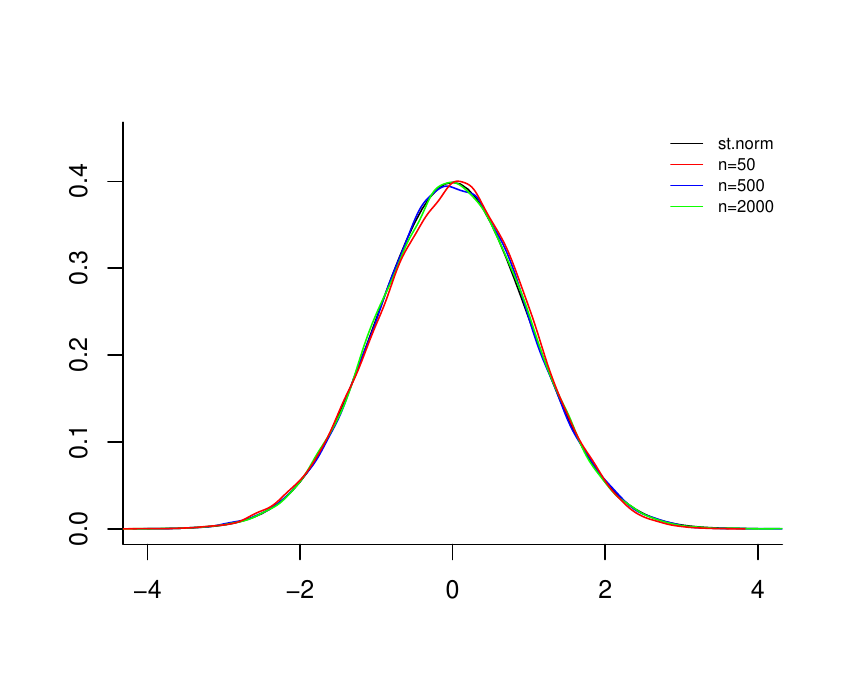}}}
\caption{ Empirical  null distribution of $st.I_n=\sqrt{n}\frac{I_n}{\widehat{\sigma}_{a,b}},$ for different $n$,  $(a, b)=(2,1)$, and marginals  $E(1.5)\times P(2)$, $N=1000$} 
\label{sampleNIn}
\end{figure}

\begin{table}[h!]
    \centering
    \resizebox{1\textwidth}{!}{
    \begin{tabular}{c|c|ccccc|ccccc|ccccc}
       &  & \multicolumn{5}{c|}{$E(1.5)\times P(2)$} & \multicolumn{5}{c|}{$E(1.5)\times NB(2,0.4)$} & \multicolumn{5}{c}{$\gamma(5,1)\times B(10,0.4)$}\\ \hline
    ${1-\alpha}$  & \diagbox{$(a,b)$}{$n$}  &  $50$ & $100$ & $500$ & $1000$ & $2000$ & $50$ & $100$ & $500$ & $1000$ & $2000$ & $50$ & $100$ & $500$ & $1000$ & $2000$\\\hline
       \multirow{3}{*}{\rotatebox[]{90}{0.95}} &
     $(2,1)$ & 1.92 & 1.96 & 1.95 & 1.96 & 1.96 & 1.93 & 1.94 & 1.95 & 1.96 &1.96 & 1.92& 1.94& 1.96 & 1.97 &  1.95\\
    & $(0.5,2)$ & 1.93 & 1.94 & 1.95 & 1.95 & 1.96 & 1.94 & 1.93 & 1.96 & 1.96 & 1.97 & 1.93 & 1.94& 1.96 & 1.97 &  1.96 \\
     &  limit &\multicolumn{15}{c}{1.96} \\ \hline\hline
%    \end{tabular}}
%    \caption{$q_{0.95}$}
%    \label{tab:q0.95}
%\end{table}

%\begin{table}[h!]
%    \centering
%    \resizebox{1\textwidth}{!}{
%   \begin{tabular}{c|ccccc|ccccc|ccccc}
%         & \multicolumn{5}{c|}{$E(1.5)\times P(2)$} & \multicolumn{5}{c|}{$E(1.5)\times NB(2,0.4)$} & \multicolumn{5}{c}{$a(5,1)\times B(10,0.4)$}\\ \hline
%       \diagbox{$(a,b)$}{$n$}  &  $50$ & $100$ & $500$ & $1000$ & $2000$ &$50$ & $100$ & $500$ & $1000$ & $2000$ & $50$ & $100$ & $500$ & $1000$ & $2000$\\\hline
       \multirow{3}{*}{\rotatebox[]{90}{0.99}} &   $(2,1)$ & 2.51 & 2.54 & 2.53 & 2.55 &2.58 & 2.55 & 2.55 & 2.54 & 2.57& 2.57 & 2.51 & 2.56 & 2.59 & 2.56 & 2.58\\
     &   $(0.5,2)$ & 2.50 & 2.53 & 2.56 & 2.57 &2.59 & 2.50 & 2.55 & 2.59 & 2.55 & 2.56 & 2.52 & 2.57 & 2.54 & 2.57 & 2.58\\
      & limit &\multicolumn{15}{c}{2.58} \\ \hline
    \end{tabular}}
    \caption{$1-\alpha$ quintiles of empirical null distribution of $\sqrt{n}\frac{\mid I_n\mid}{\widehat{\sigma}_{a,b}}$}
    \label{tab:q0.99}
\end{table}

From the presented results, we can see that convergence depends on the marginal distributions and quantile levels. For $\alpha=0.05$ the approximation is justified for $n=500$, while for smaller $\alpha$, such as $\alpha=0.01$, we need larger sample sizes to use the approximation results. For smaller $\alpha$ the test will be conservative.

%Since all projections of $\sqrt{n}\xi_n(s,t)$ are normally distributed, $\sqrt{n}\xi(s,t)$ 

\section{From bivariate to multivariate: Generalizing $T_n$ and $I_n$} \label{sec: general}
In this section, we considered natural generalizations of the bivariate setting, the cases of testing independence between two multivariate vectors, as well as the case of testing total independence for multivariate vectors with components of mixed type. For introduced statistics, we present their asymptotic properties, while the results of the power study are deferred to Section \ref{sec:power}.

\subsection{Testing independence of two random vectors}
We  consider the problem of testing the mutual independence of
 $\boldsymbol{X} \in \mathbb{R}_{+}^{r_1}$ and $\boldsymbol{Y} \in \mathbb{N}_0^{r_2}$.
 
 Let  $\boldsymbol{s}=(s_1,...,s_{r_1})\in\mathbb{R}_{+}^{r_1}$ and $\boldsymbol{t}=(t_1,...,t_{r_2})\in \mathbb{N}_0^{r_2},$ and $(\boldsymbol{X}_1,\boldsymbol{Y}_1),...,(\boldsymbol{X}_n,\boldsymbol{Y}_n )$, where $\boldsymbol{X}_i=(X_{i1},...,X_{ir_1})$ and $\boldsymbol{Y}_i=(Y_{i1},...,Y_{ir_2})$, be the available i.i.d. sample. 
Then we consider the empirical transform
\begin{align}\label{psi_multivariate}
\psi_n(\boldsymbol{s},\boldsymbol{t})&=\frac{1}{n} \sum\limits_{i=1}^{n} e^{-\sum_{j=1}^{r_1}s_jX_{ij}}\prod_{k=1}^{r_2} t_k^{Y_{ik}}.
\end{align}

 In this case, we propose the following generalizations of statistics $I_n$ and $T_n$
\begin{align*}
 I^{r_1,r_2}_{n}(w)&=\int_{\mathbb{R}_{+}^{r_1}}\int_{[0,1]^{r_2}}(\psi_{n}(\boldsymbol{s}, \boldsymbol{t}) - L_{n}(\boldsymbol{s}) G_{n}(\boldsymbol{t}))w(\boldsymbol{s},\boldsymbol{t})d\boldsymbol{s}d\boldsymbol{t},\\
  T^{r_1,r_2}_{n}(w)&=\int_{\mathbb{R}_{+}^{r_1}}\int_{[0,1]^{r_2}}(\psi_{n}(\boldsymbol{s}, \boldsymbol{t}) - L_{n}(\boldsymbol{s}) G_{n}(\boldsymbol{t}))^2w(\boldsymbol{s},\boldsymbol{t})d\boldsymbol{s}d\boldsymbol{t},
  \end{align*}
where 
\begin{align*}
L_n(\boldsymbol{s}) &=
\psi_n(\boldsymbol{s},\boldsymbol{1})=
\frac{1}{n}\sum\limits_{i=1}^{n} e^{-\sum_{j=1}^{r_1}s_jX_{ij}},\\ 
G_n(\boldsymbol{t}) &=\psi_n(\boldsymbol{0},\boldsymbol{t})= \frac{1}{n}\sum\limits_{i=1} ^{n} \prod_{k=1}^{r_2} t_k^{Y_{ik}},
\end{align*}
and  $\boldsymbol{1}$ is a vector of $r_2$ ones and $\boldsymbol{0}$ is a vector of $r_1$ zeros, {and $w:\mathbb{R}_+^{r_2}\times [0,1]^{r_1} \to \mathbb{R}_+^{(r_1+r_2)}$  is  a  weight function  such that 
$\int_{\mathbb{R}_{+}^{r_1}}\int_{[0,1]^{r_2}}w(\boldsymbol{s},\boldsymbol{t})d\boldsymbol{s}d\boldsymbol{t}<\infty$}.

For the weight function $w(\boldsymbol{s},\boldsymbol{t})=e^{-\sum_{i=1}^{r_1}a_is_i}\prod_{k=1}^{r_2}t_k^{b_{k}},\; (a_1,...,a_{r_1})\in \mathbf{R}^{r_1}_{+} $ and $(b_1,...,b_{r_2})\in \mathbf{R}^{r_2}_{+}$, the statistic ${I}^{r_1,r_2}_n$ can also be represented as
\begin{align*}
    I_n^{r_1,r_2}=&
  \frac{1}{n}\sum_{i=1}^{n} \frac{1}{\prod_{k_1=1}^{r_1}(X_{ik_1} + a_{k_1})\prod_{k_2=1}^{r_2}(Y_{ik_2} + b_{k_2}+ 1)}\\&- \frac{1}{n^2} \sum_{i=1}^{n}\frac{1}{\prod_{k_1=1}^{r_1}(X_{ik_1} + a_{k_1})}\sum\limits_{i=1}^n\frac{1}{\prod_{k_2=1}^{r_2}(Y_{i{k_2}} + b_{k_2}+1)},
  \end{align*}
  while $T_n^{r_1,r_2}$ admits following representation
  \begin{align*}
  T_n^{r_1,r_2}=&\frac{1}{n^2} \sum_{i=1}^{n} \sum_{j=1}^{n} \frac{1}{\prod_{k_1=1}^{r_1}(X_{ik_1} + X_{jk_1}+a_{k_1})\prod_{k_2=1}^{r_2}(Y_{ik_2} + Y_{jk_2} +b_{k_2}  + 1)}\\
 &+\frac{1}{n^4}\sum_{i=1}^{n}\sum_{j=1}^{n}\frac{1}{\prod_{k_1=1}^{r_1}(X_{ik_1}+X_{jk_1}+a_{k_1})}\sum_{i=1}^{n}\sum_{j=1}^{n}\frac{1}{\prod_{k_2=1}^{r_2}(Y_{ik_2}+Y_{jk_2}+b_{k_2}+1)}\\
 &-2\frac{1}{n^3}\sum_{i=1}^{n}\left( \sum_{j=1}^{n} 
     \frac{1}{\prod_{k_1=1}^{r_1}(X_{ik_1}+X_{jk_1}+a_{k_1})} 
          \sum_{q=1}^{n}\frac{1}{\prod_{k_2=1}^{r_2}(Y_{ik_2}+Y_{qk_2}+b_{k_2}+1)}\right).
\end{align*}
Now, the underlying   empirical process $\xi_n$, which is the base for assessing limiting properties, is defined by
\begin{align*}
\xi_n({\boldsymbol{s},\boldsymbol{t}})=
\psi_n(\boldsymbol{s},\boldsymbol{t})-
\psi_n(\boldsymbol{s},\boldsymbol{1})
\psi_n(\boldsymbol{0},\boldsymbol{t}).
\end{align*}
Analogously, as before, we can prove the following theorem.
\begin{theorem}\label{asym}
    Let the null hypothesis hold. Then
    \begin{itemize}
        \item $nT_n^{r_1,r_2}\overset{d}{\to}\sum_{n=1}^{\infty}\lambda_{i,\omega}Z_i^2,$ where $\lambda_{1,\omega}\geq\lambda_{2,\omega}\geq\cdots $
        us the sequence of eigenvalues of the covariance operator defined with kernel
        \begin{align*}
K((\boldsymbol{s_1},\boldsymbol{t_1}),(\boldsymbol{s_2},\boldsymbol{t_2}))=&L(\boldsymbol{s_1}+\boldsymbol{s_2})G(\boldsymbol{t_1}\boldsymbol{t_2})
  +L(\boldsymbol{s_1})L(\boldsymbol{s_2})G(\boldsymbol{t_1})G(\boldsymbol{t_2})\\&
  -
  L(\boldsymbol{s_1}+\boldsymbol{s_2})G(\boldsymbol{t_1})G(\boldsymbol{t_2})
  -L(\boldsymbol{s_1})L(\boldsymbol{s_2})G(\boldsymbol{t_1}\boldsymbol{t_2})
\end{align*}

        and an integrabile weight function $\omega$, and $Z_1, Z_2,...$ is the sequence of independent standard Gaussian random variables;

        \item $\sqrt{n}\frac{I^{r_1,r_2}_{n}}{{\sigma}_w}\overset{d}{\to}Z$, where $Z$ is standard Gaussian random variable and 
       % \label{disperzija2}
\begin{align}\label{disperzija2RV}
{\sigma}^2_w=\int_{(\mathbb{R}^+)^{r_1}}\int_{[0,1]^{r_2}}\int_{(\mathbb{R}^+)^{r_1}}\int_{[0,1]^{r_2}} K((\boldsymbol{s_1},\boldsymbol{t_1}),(\boldsymbol{s_2},\boldsymbol{t_2}))w(\boldsymbol{s_1},\boldsymbol{t_1})w(\boldsymbol{s_2},\boldsymbol{t_2})d\boldsymbol{t_1}d\boldsymbol{s_1}d\boldsymbol{t_2}d\boldsymbol{s_2}.
\end{align}
 Specifically, for        $w(\boldsymbol{s},\boldsymbol{t})=e^{\sum_{i=1}^{r_1}a_is_i}\prod_{k=1}^{r_2}t_k^{b_{k}},\; (a_1,...,a_{r_1})\in \mathbf{R}^{r_1}_{+} $ and $(b_1,...,b_{r_2})\in \mathbf{R}^{r_2}_{+}$,
\begin{align*}
 \sigma^2_{\boldsymbol{a},\boldsymbol{b}} &=%\frac{1}{(r_1+r_2)^2}
   Var\bigg(\prod_{i=1}^{r_1}\frac{1}{(X_i+a_i)}\bigg)Var\bigg(\prod_{i=1}^{r_2}\frac{1}{(Y_i+b_i+1)}\bigg).
\end{align*}
    \end{itemize}
\end{theorem}

As in the case for $r_1=r_2=1$ 
we can use a ratio-consistent estimator of $\sigma^2_{a,b}$ analogous to \eqref{sigmaab}.

\subsection{Testing total independence}
Here we consider the problem of testing total independence within the components of vector $\boldsymbol{Z}=(X_1,...,X_{r_1},Y_1,...,Y_{r_2}) \in \mathbb{R}_{+}^{r_1}\times \mathbb{N}_0^{r_2}$. 
Let $(X_{11},....,X_{1r_1},Y_{11},...,Y_{1r_{2}}),...,(X_{n1},....,X_{nr_1},Y_{n1},...,Y_{nr_{2}})$ be a multivariate sample equally distributed as $\boldsymbol{Z}$.  % Then%, for $\boldsymbol{s}=(s_1,...,s_{r_1})\in\mathbb{R}_{+}^{r_1}$ and $\boldsymbol{t}=(t_1,...,t_{r_2})\in N_0^{r_2},$ we consider empirical transform
%\begin{align*}
%\psi_n(\boldsymbol{s},\boldsymbol{t})&=\frac{1}{n} \sum\limits_{i=1}^{n} e^{-\sum_{j=1}^{r_1}%s_jX_{ij}}\prod_{k=1}^{r_2} t_k^{Y_{ik}},
%\end{align*}
%where $\boldsymbol{1}$ is a vector of $r_1$ ones and $\boldsymbol{0}$ is a vector of $r_2$ zeros.

Then, the appropriate extensions of  previously considered statistics are given by
 \begin{align*}
 \widetilde{I}^{r_1,r_2}_{n}(w)&=\int_{\mathbb{R}_{+}^{r_1}}\int_{[0,1]^{r_2}}\Big(\psi_{n}(\boldsymbol{s}, \boldsymbol{t}) -\prod_{i=1}^{r_{1}}\psi_n(\underbrace{0,...,0}_{i-1},s_i,\underbrace{0...0}_{r_1-i},\boldsymbol{1})
 \cdot\prod_{k=1}^{r_2}\psi_n(\boldsymbol{0},\underbrace{1,...,1}_{k-1},t_k,\underbrace{1,...,1}_{r_2-k})\Big)
w(\boldsymbol{s},\boldsymbol{t})d\boldsymbol{s}d\boldsymbol{t}\\
  \widetilde{T}^{r_1,r_2}_{n}(w)&=\int_{\mathbb{R}_{+}^{r_1}}\int_{[0,1]^{r_2}}\Big(\psi_{n}(\boldsymbol{s}, \boldsymbol{t}) - \prod_{i=1}^{r_{1}}\psi_n(\underbrace{0,...,0}_{i-1},s_i,\underbrace{0...0}_{r_1-i},\boldsymbol{1})
 \cdot\prod_{k=1}^{r_2}\psi_n(\boldsymbol{0},\underbrace{1,...,1}_{k-1},t_k,\underbrace{1,...,1}_{r_2-k})\Big)^2w(\boldsymbol{s},\boldsymbol{t})d\boldsymbol{s}d\boldsymbol{t},
  \end{align*}
where $w:\mathbb{R}_+^{r_2}\times [0,1]^{r_1} \to \mathbb{R}_+^{(r_1+r_2)}$  is  a  weight function  that  such that 
$\int_{\mathbb{R}_{+}^{r_1}}\int_{[0,1]^{r_2}}w(\boldsymbol{s},\boldsymbol{t})d\boldsymbol{s}d\boldsymbol{t}<\infty$,
and $\psi_n$ is introduced in \eqref{psi_multivariate}.

 After calculations, for $w(\boldsymbol{s},\boldsymbol{t})=e^{-\sum_{i=1}^{r_1}a_is_i}\prod_{k=1}^{r_2}t_k^{b_{k}},\; (a_1,...,a_{r_1})\in \mathbf{R}^{r_1}_{+} $ and $(b_1,...,b_{r_2})\in \mathbf{R}^{r_2}_{+}$, we get
\begin{align*}
\widetilde{I}^{r_1,r_2}_{n} =&\frac{1}{n}\sum\limits_{i=1}^{n}\frac{1}{\prod_{k_1=1}^{r_1}(X_{i{k_1}}+a_{k_1})\prod_{k_2=1}^{r_2}(Y_{ik_2}+b_{k_2}+1)}\\
&-\frac{1}{n^{r_1+r_2}}\prod_{k_1=1}^{r_1}\sum\limits_{i=1}^{n}\frac{1}{(X_{ik_1}+a_{k_1})}\prod_{k_2=1}^{r_2}
\sum\limits_{j=1}^{n}
\frac{1}{(Y_{jk_2}+b_{k_2}+1)}
\end{align*}
and
\begin{equation}\label{tn}
\begin{aligned}
\widetilde{T}^{r_1,r_2}_{n}
    =& \frac{1}{n^2} \sum_{i=1}^{n} \sum_{j=1}^{n} \frac{1}{\prod_{k_1=1}^{r_1}(X_{ik_1} + X_{jk_1}+a_{k_1})\prod_{k_2=1}^{r_2}(Y_{ik_2} + Y_{jk_2} +b_{k_2}  + 1)}\\
    &+\frac{1}{n^{2(r_1+r_2)}}\prod_{k_1=1}^{r_1}\sum_{i=1}^{n} \sum_{j=1}^{n}\frac{1}{(X_{ik_1} + X_{jk_1}+a_{k_1})}\prod_{k_2=1}^{r_2}\sum_{i=1}^{n}\sum_{j=1}^{n}\frac{1}{(Y_{ik_2} + Y_{jk_2} + b_{k_2}  + 1)}\\
    &-2\frac{1}{n^{r_1+r_2+1}} \sum_{i=1}^{n} \Bigg(\sum_{l_1=1}^n \cdots \sum_{l_{r_1}=1}^{n}
    \sum_{m_1=1}^n \cdots \sum_{m_{r_2}=1}^{n}
    \frac{1}{\prod_{j=1}^{r_1}(X_{ij} + X_{l_{j}j}+a_{j})\prod_{k=1}^{r_2}(Y_{ik} + Y_{m_{k}k}  +b_{k} + 1)}\Bigg).
 \end{aligned}   
\end{equation}

\begin{remark}
{A more efficient way for implementation is to represent the third summand of \eqref{tn} as
\begin{align*}
  %&-2\frac{1}{n^{r_1+r_2+1}} \sum_{i=1}^{n} \Bigg(\sum_{l_1=1}^n \cdots \sum_{l_{r_1}=1}^{n}
   % \sum_{m_1=1}^n \cdots 
   % \sum_{m_{r_2}=1}^{n}
   % \frac{1}{\prod_{j=1}^{r_1}(X_{ij} + X_{l_{j}j}+a_{j})\prod_{k=1}^{r_2}(Y_{ik} + Y_{m_{k}k}  +b_{k} + 1)}\Bigg)\\
    %=&-2\frac{1}{n^{r_1+r_2+1}} \sum_{i=1}^{n} \Bigg(\sum_{l_1=1}^n \cdots %\sum_{l_{r_1}=1}^{n}
    %\frac{1}{\prod_{j=1}^{r_1}(X_{ij} + X_{l_{j}j}+a_{j})}
    %\sum_{m_1=1}^n \cdots 
    %\sum_{m_{r_2}=1}^{n}
    %\frac{1}{\prod_{k=1}^{r_2}(Y_{ik} + Y_{m_{k}k}  +b_{k} + 1)}\Bigg)\\
      %=
      &-2\frac{1}{n^{r_1+r_2+1}} \sum_{i=1}^{n} \Bigg(\prod_{j=1}^{r_1}\sum_{l_j=1}^n
    \frac{1}{(X_{ij} + X_{l_{j}j}+a_{j})}
    \prod_{k=1}^{r_2}
    \sum_{m_k=1}^n
    \frac{1}{(Y_{ik} + Y_{m_{k}k}  +b_{k} + 1)}\Bigg).
\end{align*}}

\end{remark}

The limiting results can be derived analogously to those obtained for the previously studied bivariate vectors, by analyzing the empirical process defined by
\begin{align*}
    \widetilde{\xi}_n(\boldsymbol{s},\boldsymbol{t})&=
    \frac{1}{n} \sum\limits_{i=1}^{n} \Big( e^{-\sum_{j=1}^{r_1}s_jX_{ij}}\prod_{k=1}^{r_2} t_k^{Y_{ik}}\Big)
    -
    \frac{1}{n^{r_1}}\prod_{j=1}^{r_1} \sum_{i=1}^{n} \Big(e^{- s_j X_{ij}}\Big)
    \frac{1}{n^{r_2}} \prod_{j=1}^{r_2}\sum_{i=1}^{n} \Big( t_j^{Y_{ij}}\Big)\\
    &=\psi_n(\boldsymbol{s},\boldsymbol{t})-\prod_{j=1}^{r_1}L_n(s_j)\prod_{k=1}^{r_2}G_n(t_k)
\end{align*}
from where we can derive null distributions of both statistics.
Therefore we present results without proof.
\begin{theorem}\label{teoremaTI}
    Let the null hypothesis hold. Then 
    \begin{itemize}
        \item $n\widetilde{T}^{r_1,r_2}_{n}\overset{d}{\to}\sum\limits_{n=1}^\infty \widetilde{\lambda}_{i,\omega}Z_i^2$, where $\widetilde{\lambda}_{1,\omega}\geq \widetilde{\lambda}_{2,\omega}\geq ...$ is the sequence of eigenvalues of the covariance operator defined with kernel 
        \begin{align*}
%K(\boldsymbol{s_1},\boldsymbol{s_2},\boldsymbol{t_1},\boldsymbol{t_2})=
%  & \frac{1}{(r_1+r_2)^2}\Bigg[\prod_{j=1}^{r_1}L(s_{1j}+s_{2j})\prod_{k=1}^{r_2}G(t_{1k}t_{2k})-\prod_{j=1}^{r_1} L(s_{1j})L(s_{2j})\prod_{k=1}^{r_2}G(t_{1k})G(t_{2k})\\
% &-\sum_{d=1}^{r_1}[L(s_{1d}+s_{2d})-L(s_{1d})L(s_{2d})]\prod_{j\neq d,j=1}^{r_1}
 %L(s_{1j})L(s_{2j}) \prod_{k=1}^{r_2}G(t_{1k})G(t_{2k})\\
% &-\sum_{d=1}^{r_2}[G(t_{1d}t_{2d})-G(t_{1d})G(t_{2d})]\prod_{j=1}^{r_1}L(s_{1j})L(s_{2j})\prod_{k\neq d, k=1}^{r_2}G(t_{1k})G(t_{2k})\Bigg]\\
% or\\
%\end{align*}
\widetilde{K}((\boldsymbol{s_1},\boldsymbol{t_1}),(\boldsymbol{t_2},\boldsymbol{s_2}))=&
  \prod_{j=1}^{r_1}L(s_{1j}+s_{2j})\prod_{k=1}^{r_2}G(t_{1k}t_{2k})\\&+(r_1+r_2-1)\prod_{j=1}^{r_1} L(s_{1j})L(s_{2j})\prod_{k=1}^{r_2}G(t_{1k})G(t_{2k})\\
 &-\sum_{d=1}^{r_1}L(s_{1d}+s_{2d})\prod_{j\neq d,j=1}^{r_1}
 L(s_{1j})L(s_{2j})\prod_{k=1}^{r_2}G(t_{1k})G(t_{2k})\\
 &-\prod_{j=1}^{r_1}L(s_{1j})L(s_{2j})\sum_{d=1}^{r_2}G(t_{1d}t_{2d})\prod_{k\neq d, k=1}^{r_2}G(t_{1k})G(t_{2k}),
\end{align*}  
        and an integrabile weight function $w$, and $Z_1, Z_2,...$ is the sequence of independent standard Gaussian random variables;
        \item $\sqrt{n}\frac{\widetilde{I}^{r_1,r_2}_{n}}{\widetilde{\sigma}_w}\overset{d}{\to}Z$, where $Z$ is standard Gaussian random variable and 
\begin{align}\label{disperzija1}
    \widetilde{\sigma}^2_w=\int_{(\mathbb{R}^+)^{r_1}}\int_{[0,1]^{r_2}}\int_{(\mathbb{R}^+)^{r_1}}\int_{[0,1]^{r_2}} \widetilde{K}((\boldsymbol{s_1},\boldsymbol{t_1}),(\boldsymbol{s_2},\boldsymbol{t_2}))w(\boldsymbol{s_1},\boldsymbol{t_1})w(\boldsymbol{s_2},\boldsymbol{t_2})d\boldsymbol{t_1}d\boldsymbol{s_1}d\boldsymbol{t_2}d\boldsymbol{s_2}.
\end{align}
        Specially, for        $w(\boldsymbol{s},\boldsymbol{t})=e^{-\sum_{i=1}^{r_1}a_is_i}\prod_{k=1}^{r_2}t_k^{b_{k}}$
{
\begin{align}
\begin{split}
\label{disperzija2TI}
     %   \sigma^2_{\boldsymbol{a},\boldsymbol{b}}
     %  =& \frac{1}{(r_1+r_2)^2}\Bigg[\prod_{j=1}^{r_1}
  %E\Big(\frac{1}{X_{j}+a_{j}}\Big)^2
  %\prod_{k=1}^{r_2}
  %E\Big(\frac{1}{Y_{k}+b_{k}+1}\Big)^2\\
  %&-\prod_{j=1}^{r_1} E^2\Big(\frac{1}{X_{j}+a_{j}}\Big) \prod_{k=1}^{r_2}E^2\Big(\frac{1}{Y_{k}+b_{k}+1}\Big)\\
 %&-\sum_{d=1}^{r_1}Var\Big(\frac{1}{X_{d}+a_{d}}\Big)\prod_{j\neq d,j=1}^{r_1}
 %E^2\Big(\frac{1}{X_{j}+a_{j}}\Big)
 %\prod_{k=1}^{r_2}
 %E^2\Big(\frac{1}{Y_{k}+b_{k}+1}\Big)\\
 % &-\sum_{d=1}^{r_2}Var\Big(\frac{1}{Y_{d}+a_{d}}\Big)\prod_{k\neq d,k=1}^{r_2}
 %E^2\Big(\frac{1}{Y_{k}+a_{k}}\Big)
 %\prod_{j=1}^{r_1}
 %E^2\Big(\frac{1}{X_{j}+b_{j}+1}\Big)\Bigg]\\
\widetilde{\sigma}^2_{\boldsymbol{a},\boldsymbol{b}}
       =&\prod_{j=1}^{r_1}
  E\Big(\frac{1}{X_{j}+a_{j}}\Big)^2
  \prod_{k=1}^{r_2}
  E\Big(\frac{1}{Y_{k}+b_{k}+1}\Big)^2\\
  &+(r_1+r_2-1)\prod_{j=1}^{r_1} \Big(E\Big(\frac{1}{X_{j}+a_{j}}\Big)\Big)^2 \prod_{k=1}^{r_2}\Big(E\Big(\frac{1}{Y_{k}+b_{k}+1}\Big)\Big)^2\\
 &-\sum_{d=1}^{r_1}E\Big(\frac{1}{X_{d}+a_{d}}\Big)^2\prod_{j\neq d,j=1}^{r_1}
 \Big(E\Big(\frac{1}{X_{j}+a_{j}}\Big)\Big)^2
 \prod_{k=1}^{r_2}
 \Big(E\Big(\frac{1}{Y_{k}+b_{k}+1}\Big)\Big)^2\\
  &-\sum_{d=1}^{r_2}E\Big(\frac{1}{Y_{d}+b_{d}+1}\Big)^2\prod_{k\neq d,k=1}^{r_2}
 \Big(E\Big(\frac{1}{Y_{k}+b_{k}+1}\Big)\Big)^2
 \prod_{j=1}^{r_1}
 \Big(E\Big(\frac{1}{X_{j}+a_{j}}\Big)\Big)^2.
\end{split}
\end{align}}
    \end{itemize}
\end{theorem}
As in the case for $r_1=r_2=1$ 
we can use a ratio-consistent estimator of $\widetilde{\sigma}^2_{a,b}$ analogous to \eqref{sigmaab}.

\section{Power study}\label{sec:power}

Here we explore the power performance of the proposed tests and compare them with the power of a competitor test from \cite{GFL19}. The warp-speed modification from \cite{GPW} was used with 10,000 replicates for both small sample sizes $(n=20)$ and moderate sample sizes $(n=50)$, with the  level of significance  $\alpha = 5\%$.

The competitor test  from \cite{GFL19}, denoted by $D_n$, is based on the weighted difference between the joint empirical characteristic function and the product of the marginal empirical characteristic functions
%\begin{align*}
%    D_n^r(v)= n \int_{R^r} \| C_n(\boldsymbol{t})-\prod_{i=1}^r C_n(t_i)\| ^2v(\boldsymbol{t})d\boldsymbol{t}.
%\end{align*}
\begin{align*}
 D_n^{r_1,r_2}(v) &= n \int_{\mathbb{R}_{+}^{r_1}}\int_{[0,1]^{r_2}}
\bigg( C_n(\boldsymbol{s}, \boldsymbol{t}) - C_n(\boldsymbol{s}) \cdot C_n(\boldsymbol{t}) \bigg)^2 v(\boldsymbol{s}, \boldsymbol{t})\, d\boldsymbol{s}\, d\boldsymbol{t},\\
\widetilde{D}_n^{r_1,r_2}(v)&= n \int_{\mathbb{R}_{+}^{r_1}}\int_{[0,1]^{r_2}}
 \left( C_n(\boldsymbol{s},\boldsymbol{t})-\prod_{i=1}^{r_1} C_n(s_i)\prod_{k=1}^{r_2} C_n(t_k)\right) ^2v(\boldsymbol{s},\boldsymbol{t})d\boldsymbol{s}d\boldsymbol{t}.
\end{align*}
%Total - vector $\boldsymbol{Z}$ and $r=r_1+r_2$ 
%\begin{align*}
%\widetilde{D}^r_{n}(v) =&  
%\frac{1}%{n}\sum\limits_{j,k=1}^{n}e^{-0.125\sum_{d=1}^r(Z_{jd}-%Z_{kd})^2}+
%\frac{1}{n^{2r-1}}\prod_{d=1}^r
%\sum\limits_{j,k=1}^{n}e^{-0.125(Z_{jd}-Z_{kd})^2}\\
%-&
%\frac{2}{n^r}\sum\limits_{i_1,i_2,...,i_d, %v=1}^{n}e^{-0.125\sum_{r=1}^d(Z_{i_r}-Z_v)^2}
%\end{align*}
% Two vectors - $\boldsymbol{X}$ $r_1$ and% $\boldsymbol{Y}$  $r_2$
%\begin{align*}
%D_n^r(v)=&\frac{1}{n^2} \sum_{j,k=1}^{n} e^{-0.125 %\big(\sum_{i=1}^{r1} (X_{ji} - X_{ki})^2+\sum_{i=1}^{r2} (Y_{ji} - Y_{ki})^2\big)}\\
%&+\frac{1}{n^4} \sum_{j,m=1}^{n} e^{-0.125 \cdot \sum_{i=1}^{r1} (X_{ji} - X_{mi})^2}  \sum_{k,l=1}^{n} e^{-0.125\cdot \sum_{i=1}^{r2} (Y_{ki} - Y_{li})^2}\\
%&+\frac{2}{n^3} \sum_{j=1}^{n}  \Big(\sum_{k=1}^{n}e^{-0.125 \cdot  \sum_{i=1}^{r1} (X_{ji} - X_{ki})^2}
%\sum_{l=1}^{n}e^{ -0.125 \cdot \sum_{i=i}^{r2} (Y_{ji} %- Y_{li})^2 }\Big)
%\end{align*}}}

In accordance with the authors' recommendation, we use the weight function 
$v(\boldsymbol{s},\boldsymbol{t})=e^{-0.5\sum_{i=1}^{r_1}\sigma_{i1}^2s_i^2-0.5\sum_{i=1}^{r_2}\sigma_{i2}^2t_i^2}$ and $\sigma_{i1}^2=\sigma_{i2}^2=0.5^2$.

Following \cite{O16mixed}, we consider Gamma distributions $(\gamma(\alpha,\beta))$ for continuous margins (in particular, the notation $E(\beta)$ is used when $\alpha = 1$, i.e. when it is exponential distribution), and Poisson $(P)$, Binomial $(B(N,p))$, and Negative binomial $(NB(m,p))$ distributions for discrete margins.
Their dependence structure is modeled using R-vine copulas, with the following families: Gaussian ($Ga(\theta)$) \cite{JWWD16, LW11, KAD21, BC21}, Clayton ($Cl(\theta)$) \cite{KAD21, BC21, YL22}, Gumbel ($Gu(\theta)$) \cite{YL22, ZK16, HJ12}, and Joe ($Joe(\theta)$) \cite{YL22, HJ12, KR22}. In Tables \ref{power}-\ref{powerM}, Ind. refers to the model with selected margins under the null hypothesis of independence. In multivariate settings when testing independence of two random vectors or total independence, used  R-vine copula structures are presented on Figures \ref{fig:12}, \ref{fig:21} and \ref{fig:61}.  The parameter $\theta_1$ corresponds to the (conditional) association between two variables of the same type at each of the considered levels, whereas the parameters $\theta_2$ and $\theta_3$ describe the (conditional) associations between variables of different types.

\begin{figure}[h!]
  \centering
   \includegraphics[width=0.4\textwidth]{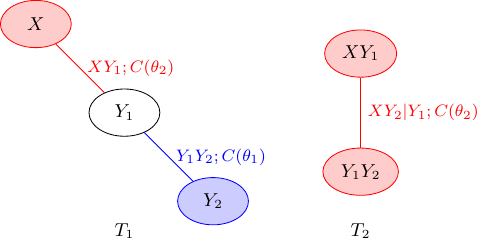}
  \caption{{{Graphical representation of R-vine copula structures $T_1-T_2$ modeling dependencies between variables $X$ and $(Y_1,Y_2)$, with different copula families applied at successive dependency levels}
%  R-vine tree plots with copula families $C$ %=\{Ga, Cl, Gu, Joe\}$ 
 % and parameters $\theta$ %=\{\theta_1,\theta_2,\theta_3\}$ 
  %  for three continuous and three discrete random variables in the case of total independence. For the case of two vectors, $\theta_1$ is replaced by $\theta_3$.
  }}
  
  \label{fig:12}
\end{figure}

\begin{figure}[h!]
  \centering
   \includegraphics[width=0.4\textwidth]{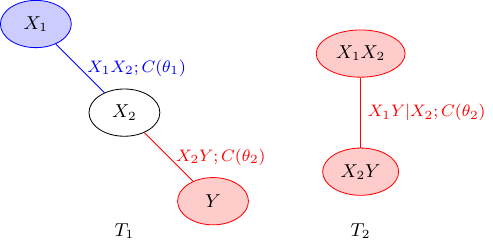}
  \caption{{{Graphical representation of R-vine copula structures $T_1-T_2$ modeling dependencies between variables $(X_1,X_2)$ and $Y$, with different copula families applied at successive dependency levels}
%  R-vine tree plots with copula families $C$ %=\{Ga, Cl, Gu, Joe\}$ 
 % and parameters $\theta$ %=\{\theta_1,\theta_2,\theta_3\}$ 
  %  for three continuous and three discrete random variables in the case of total independence. For the case of two vectors, $\theta_1$ is replaced by $\theta_3$.
  }}
  
  \label{fig:21}
\end{figure}

\begin{figure}[h!]
  \centering
   \includegraphics[width=0.3\textwidth]{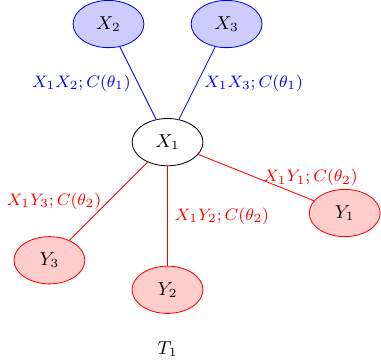}\hspace{1cm}
   \includegraphics[width=0.3\textwidth]{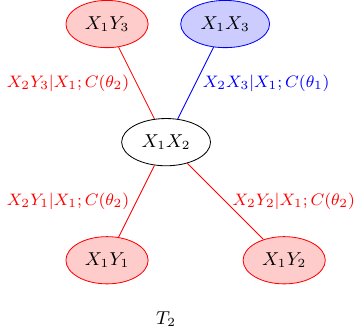}\\
   \includegraphics[width=0.3\textwidth]{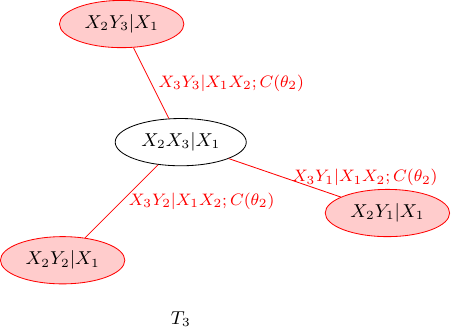}\hspace{1cm}
   \includegraphics[width=0.2\textwidth]{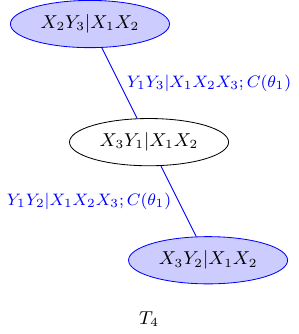}\hspace{1cm}
   \includegraphics[width=0.2\textwidth]{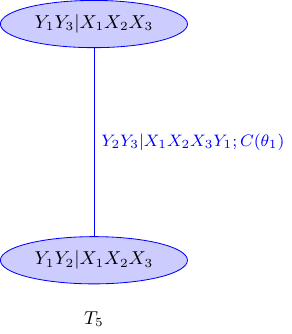}
  \caption{{{Graphical representation of R-vine copula structures $T_1-T_5$ modeling dependencies between the components of vector $(X_1,X_2,X_3,Y_1,Y_2,Y_3)$, with different copula families applied at successive dependency levels}
%  R-vine tree plots with copula families $C$ %=\{Ga, Cl, Gu, Joe\}$ 
 % and parameters $\theta$ %=\{\theta_1,\theta_2,\theta_3\}$ 
  %  for three continuous and three discrete random variables in the case of total independence. For the case of two vectors, $\theta_1$ is replaced by $\theta_3$.
  }}
  
  \label{fig:61}
\end{figure}

%|For clarity, for each alternative, we highlight in grey the values that are close to the highest value.  Therefore, grey cells correspond to power values that are less than or equal to $1.96 \times 0.5/\sqrt{10\;000}=0.01$  away from the highest.

{{For clarity, we highlight power values that are less than or equal to 
$1.96 \times 0.5/\sqrt{10,\!000} = 0.01$ away from the highest. 
In the overall comparison across all tests, dark gray with bold
indicates the highest values, while dark gray without bold
marks values that are within $1\%$ of the maximum. 
For comparisons within individual tests, only light gray is used 
to highlight both the maximum value and those within $1\%$ of it.}}

When it comes to the comparison of different statistics included in the study, from Tables \ref{power}-\ref{powerM}, we can see that $st.I_n$ performs the best, further confirming the advantage of having a simple asymptotic distribution. The competing statistic $D_n$ is never the best-performing.

We can see that the impact of the tuning parameters $\boldsymbol{a}=(\underbrace{a,\cdots,a}_{r_1})$ and $\boldsymbol{b}=(\underbrace{b,\cdots,b}_{r_2})$ is significant. In the bivariate case, for associations modeled by the Clayton copula, smaller values of $a$ and $b$ are more appropriate. Medium values are most suitable for Gaussian copulas, whereas larger values work better for the Gumbel and Joe copulas. Overall, if no prior knowledge about the data is available, we recommend choosing $a = 1$ and $b = 5$. As the dimension increases, the previously observed conclusions can no longer be generalized. In particular, the influence of tuning parameters becomes more evident, and their optimal selection depends not only on the dimensionality but also on the underlying dependence structure. Moreover, it becomes clear that the optimal choice may differ depending on whether the goal is to test independence between vectors or total independence. As a practical compromise, in the absence of prior information about the data, we recommend using $a=1$ and $b=5$ for testing independence between vectors, and $a=1$ and $b=1$ for testing total independence. If there is prior knowledge about the dependence, one can consult Tables \ref{powerMV12_high_tau}-\ref{powerM},

 %When it comes to the comparison of powers of tests for testing independence and total independence across the same settings, as expected, the test for testing total independence are more powerful.

{
It is also of interest to explore the effect of association among the components within each vector (inter-association) on the performance of tests for assessing independence between two vectors. Based on the results presented in this manuscript, as well as additional findings not included here, we conclude that inter-association tends to have a stronger negative impact for larger sample sizes; specifically, as inter-association increases, the power of the tests decreases. This effect becomes even more noticeable with increasing vector dimensions.

}

%XXX are presented in Supplementary material.

%, but also never the worst.

\begin{table}[!h]
    \centering 
     \resizebox{0.9\textwidth}{!}{
\begin{tabular}{c|c|r|r|r|r|r|r|r|r|r|r|r|r|r|r|r|r|r|r}
        \multicolumn{2}{c|}{}& \multicolumn{6}{c|}{$E(1.5)\times P(2)$} & \multicolumn{6}{c|}{$E(1.5)\times NB(2,0.4)$} & \multicolumn{6}{c}{$\gamma(5,1)\times B(10,0.4)$}\\ %\hline
        \cline{2-20}
          \multicolumn{1}{c}{} &Copula & $Ind.$  &  \multicolumn{2}{c|}{$Ga$} & $Cl$ & ${Gu}$ & ${Joe}$& $Ind.$   &  \multicolumn{2}{c|}{${Ga}$} & ${Cl}$ & ${Gu}$ & ${Joe}$& $Ind.$  &  \multicolumn{2}{c|}{${Ga}$} & ${Cl}$ & ${Gu}$ & ${Joe}$ \\ %\hline
          \cline{2-20}
         \multicolumn{1}{c}{Test} & \diagbox{$(a,b)$}{$\theta$} && 0.35 & 0.55 & 0.5 & 1.5 & 2 && 0.35 & 0.55 & 0.5 & 1.5 & 2 && 0.35 & 0.55 & 0.5 & 1.5 & 2 \\ \hline\hline
         \multicolumn{2}{c}{}& \multicolumn{16}{c}{$n=20$} \\ \hline
     \rotatebox[]{90}{$D_n$} & ~ & 5 & 23 & 46 & 14 & 46 & 60 & 5 & 20 & 42 & 13 & 43 & 53 & 5 & 22 & 51 & 20 & 46 & 55 \\ \hline
       \multirow{ 12}{*}{\rotatebox[]{90}{$T_n$}}  & (0.2, 0.5) & 6 & 27 & 52 & \cellcolor{gray!10}{33} & 39 & 34 & 6 & 29 & 56 & \cellcolor{gray!10}{33} & 41 & 36 & 6 & 23 & 43 & 31 & 34 & 26 \\ 
        ~ & (0.2, 1) & 6 & 27 & 53 & \cellcolor{gray!10}{33} & 41 & 37 & 6 & 30 & 58 & \cellcolor{gray!10}{33} & 43 & 39 & 6 & 24 & 46 & \cellcolor{gray!10}{32} & 36 & 29 \\
        ~ & (0.2, 5) & 6 & 28 & 57 & \cellcolor{gray!10}{32} & 46 & 46 & 6 & \cellcolor{gray!10}{31} & 61 & 31 & 48 & 49 & 6 & 28 & 52 & \cellcolor{gray!10}{33} & 40 & 35 \\
\hline
    \end{tabular}}
    \caption{Rejection rates for testing independence of one continuous and one discrete random variable ($\alpha=0.05,$ $ n=20, n=50$) }
    \label{power}
\end{table}

\begin{table}[!h]
    \centering 
     \resizebox{0.9\textwidth}{!}{
    \begin{tabular}{c|c|r|r|r|r|r|r|r|r|r|r|r|r|r|r|r|r|r|r}
        \multicolumn{2}{c|}{}& \multicolumn{6}{c|}{$E(1.5)\times P(2)\times P(2)$} & \multicolumn{6}{c|}{$E(1.5)\times NB(2,0.4)\times NB(2,0.4)$} & \multicolumn{6}{c}{$\gamma (5,1)\times B(10,0.4)\times B(10,0.4)$}\\ %\hline
          \cline{2-20}
          \multicolumn{1}{c}{} &Copula & $Ind.$  &  \multicolumn{2}{c|}{$Ga$} & $Cl$ & ${Gu}$ & ${Joe}$& $Ind.$   &  \multicolumn{2}{c|}{${Ga}$} & ${Cl}$ & ${Gu}$ & ${Joe}$& $Ind.$  &  \multicolumn{2}{c|}{${Ga}$} & ${Cl}$ & ${Gu}$ & ${Joe}$ \\ %\hline
                 \cline{2-20}
           & \multicolumn{1}{r|}{$\theta_1$}& & 0.75 & 0.75 & 0.75 & 4.5 & 6 & &0.75 & 0.75 & 0.75 & 4.5 & 6 && 0.75 & 0.75 & 0.75 & 4.5 & 6 \\
        \multicolumn{1}{c}{Test} & \diagbox{$(a,b)$}{$\theta_2$} && 0.35 & 0.55 & 0.5 & 1.5 & 2 && 0.35 &0.55 & 0.5 & 1.5 & 2 && 0.35 & 0.55 & 0.5 & 1.5 & 2 \\
         \hline\hline
         \multicolumn{2}{c}{}& \multicolumn{16}{c}{$n=20$} \\ \hline
           \rotatebox[]{90}{$D_n$} &  ~ & 4 & 30 & 60 & 16 & 43 & 64 & 4 & 23 & 50 & 13 & 47 & 56 & 5 & 31 & 73 & 28 & 54 & 67 \\ \hline
         \multirow{ 12}{*}{\rotatebox[]{90}{$T_n$}} & (0.2, 0.5) & 6 & 38 & 64 & 42 & 33 & 48 & 6 & 38 & 68 & 44 & 46 & 47 & 6 & 28 & 52 & 33 & 35 & 40 \\ 
        ~ & (0.2,1) & 6 & 40 & 68 & 43 & 35 & 52 & 6 & 40 & 71 & 45 & 50 & 52 & 6 & 30 & 55 & 35 & 38 & 43 \\
\hline
    \end{tabular}}
     \caption{Rejection rates for testing independence of two vectors with one continuous and two discrete components ($\alpha = 0.05$, $n = 20$, $n = 50$)}
    \label{powerMV12_high_tau}
\end{table}

\begin{table}[!h]
    \centering 
     \resizebox{0.9\textwidth}{!}{
    \begin{tabular}{c|c|r|r|r|r|r|r|r|r|r|r|r|r|r|r|r|r|r|r}
        \multicolumn{2}{c|}{}& \multicolumn{6}{c|}{$E(1.5)\times E(1.5)\times P(2)$} & \multicolumn{6}{c|}{$E(1.5)\times E(1.5)\times NB(2,0.4)$} & \multicolumn{6}{c}{$\gamma (5,1)\times \gamma(5,1)\times B(10,0.4)$}\\ %\hline
          \cline{2-20}
          \multicolumn{1}{c}{} &Copula & $Ind.$  &  \multicolumn{2}{c|}{$Ga$} & $Cl$ & ${Gu}$ & ${Joe}$& $Ind.$   &  \multicolumn{2}{c|}{${Ga}$} & ${Cl}$ & ${Gu}$ & ${Joe}$& $Ind.$  &  \multicolumn{2}{c|}{${Ga}$} & ${Cl}$ & ${Gu}$ & ${Joe}$ \\ %\hline
                 \cline{2-20}
           & \multicolumn{1}{r|}{$\theta_1$}& & 0.75 & 0.75 & 0.75 & 4.5 & 6 & &0.75 & 0.75 & 0.75 & 4.5 & 6 && 0.75 & 0.75 & 0.75 & 4.5 & 6 \\
        \multicolumn{1}{c}{Test} & \diagbox{$(a,b)$}{$\theta_2$} && 0.35 & 0.55 & 0.5 & 1.5 & 2 && 0.35 &0.55 & 0.5 & 1.5 & 2 && 0.35 & 0.55 & 0.5 & 1.5 & 2 \\
         \hline\hline
         \multicolumn{2}{c}{}& \multicolumn{16}{c}{$n=20$} \\ \hline
           \rotatebox[]{90}{$D_n$} & ~ & 4 & 34 & 68 & 21 & 58 & 69 & 4 & 31 & 61 & 19 & 53 & 62 & 4 & 30 & 68 & 27 & 52 & 63 \\ \hline
         \multirow{ 12}{*}{\rotatebox[]{90}{$T_n$}} & (0.2, 0.5) & 5 & 35 & 63 & \cellcolor{gray!10}{43} & 43 & 43 & 5 & 35 & 65 & 41 & 46 & 46 & 5 & 29 & 52 & 37 & 37 & 36 \\ 
        ~ & (0.2, 1) & 5 & 36 & 65 & 42 & 45 & 46 & 5 & 36 & 67 & 40 & 47 & 49 & 6 & 31 & 54 & 38 & 39 & 39 \\
\hline
    \end{tabular}}
     \caption{Rejection rates for testing independence of two vectors with two continuous and one discrete components ($\alpha = 0.05$, $n = 20$, $n = 50$)}
    \label{powerMV21_high_tau}
\end{table}

\begin{table}[!h] 
    \centering 
     \resizebox{0.9\textwidth}{!}{
    \begin{tabular}{c|c|r|r|r|r|r|r|r|r|r|r|r|r|r|r|r|r|r|r}
        \multicolumn{2}{c|}{}& \multicolumn{6}{c|}{$E(1.5)^3\times P(2)^3$} & \multicolumn{6}{c|}{$E(1.5)^3\times NB(2,0.4)^3$} & \multicolumn{6}{c}{$\gamma(5,1)^3\times B(10,0.4)^3$}\\ %\hline
        \cline{2-20}
          \multicolumn{1}{c}{} &Copula & $Ind.$  &  \multicolumn{2}{c|}{$Ga$} & $Cl$ & ${Gu}$ & ${Joe}$& $Ind.$   &  \multicolumn{2}{c|}{${Ga}$} & ${Cl}$ & ${Gu}$ & ${Joe}$& $Ind.$  &  \multicolumn{2}{c|}{${Ga}$} & ${Cl}$ & ${Gu}$ & ${Joe}$ \\ %\hline
          \cline{2-20}

            & \multicolumn{1}{r|}{$\theta_1$} &&  0.75 & 0.75 & 0.75 & 4.5 & 6 & & 0.75 & 0.75 & 0.75 & 4.5 & 6 && 0.75 & 0.75 & 0.75 & 4.5 & 6  \\
         \multicolumn{1}{c}{Test} & \diagbox{$(a,b)$}{$\theta_2$} & & 0.35 & 0.55 & 0.5 & 1.5 & 2 && 0.35 & 0.55 & 0.5 & 1.5 & 2 && 0.35 & 0.55 & 0.5 & 1.5 & 2\\ \hline\hline
         \multicolumn{2}{c}{}& \multicolumn{16}{c}{$n=20$} \\ \hline
     \rotatebox[]{90}{$D_n$} &             & 3   & 39           & 76           & 29      & 57        & 69     & 3   & 28            & 64            & 17       & 46         & 59      & 3   & 34           & 79           & 32      & 54        & 68     \\ \hline
\multirow{ 12}{*}{\rotatebox[]{90}{$T_n$}}    & (0.2, 0.5) & 5   & 35           & 59           & 44      & 38        & 40     & 5   & 36            & 59            & 43       & 39         & 42      & 5   & 25           & 39           & 32      & 24        & 32     \\
                         & (0.2, 1)   & 5   & 37           & 62           & 45      & 39        & 43     & 5   & 37            & 61            & 44       & 40         & 44      & 5   & 28           & 43           & 34      & 27        & 35     \\
\hline
    \end{tabular}}
    \caption{Rejection rates for testing independence of two vectors with three continuous and three discrete components ($\alpha = 0.05$, $n = 20$, $n = 50$)}
    \label{powerMV3}
\end{table}

\begin{table}[!h]
    \centering 
     \resizebox{0.9\textwidth}{!}{
    \begin{tabular}{c|c|r|r|r|r|r|r|r|r|r|r|r|r|r|r|r|r|r|r}
        \multicolumn{2}{c|}{}& \multicolumn{6}{c|}{$E(1.5)\times P(2)\times P(2)$} & \multicolumn{6}{c|}{$E(1.5)\times NB(2,0.4)\times NB(2,0.4)$} & \multicolumn{6}{c}{$\gamma (5,1)\times B(10,0.4)\times B(10,0.4)$}\\ %\hline
          \cline{2-20}
          \multicolumn{1}{c}{} &Copula & $Ind.$  &  \multicolumn{2}{c|}{$Ga$} & $Cl$ & ${Gu}$ & ${Joe}$& $Ind.$   &  \multicolumn{2}{c|}{${Ga}$} & ${Cl}$ & ${Gu}$ & ${Joe}$& $Ind.$  &  \multicolumn{2}{c|}{${Ga}$} & ${Cl}$ & ${Gu}$ & ${Joe}$ \\ %\hline
                 \cline{2-20}
          & \multicolumn{1}{|r|}{$\theta_1$}&& 0 & 0 & 0 & 0 & 0 & &0 &0 &  & 0 & 0 && 0 & 0 & 0 & 0 & 0 \\
        \multicolumn{1}{c}{Test} & \diagbox{$(a,b)$}{$\theta_2$} && 0.35 & 0.55 & 0.5 & 1.5 & 2 && 0.35 &0.55 & 0.5 & 1.5 & 2 && 0.35 & 0.55 & 0.5 & 1.5 & 2 \\
         \hline\hline
         \multicolumn{2}{c}{}& \multicolumn{16}{c}{$n=20$} \\ \hline
        \rotatebox[]{90}{$\widetilde{D}_n$} & ~ & 4 & 14 & 24 & 5 & \cellcolor{gray!35}{99} & 35 & 5 & 9 & 12 & 6 & 14 & 16 & 6 & 25 & 62 & 21 & 54 & 63 \\ \hline
        \multirow{ 12}{*}{\rotatebox[]{90}{$\widetilde{T}_n$}} & (0.2, 0.5) & 5 & 33 & 58 & \cellcolor{gray!35}{\bf 28} & 87 & 42 & 5 & 32 & 55 & 35 & 44 & 40 & 5 & 29 & 50 & 32 & 38 & 33 \\ 
        ~ & (0.2,1) & 5 & 37 & 63 & \cellcolor{gray!35}{27} & 87 & 49 & 5 & 36 & 61 & 37 & 49 & 47 & 5 & 32 & 54 & 34 & 41 & 38 \\
\hline        
    \end{tabular}}
     \caption{Rejection rates for testing total independence of random vector with one continuous and two discrete components ($\alpha = 0.05$, $n = 20$, $n = 50$)}

    % Rejection rates for testing independence of one continuous and one discrete random variable ($\alpha=0.05,$ $ n=20, n=50$)
    \label{powerMV12_tot}
\end{table}

\begin{table}[!h]
    \centering 
     \resizebox{0.9\textwidth}{!}{
    \begin{tabular}{c|c|r|r|r|r|r|r|r|r|r|r|r|r|r|r|r|r|r|r}
        \multicolumn{2}{c|}{}& \multicolumn{6}{c|}{$E(1.5)\times E(1.5)\times P(2)$} & \multicolumn{6}{c|}{$E(1.5)\times E(1.5)\times NB(2,0.4)$} & \multicolumn{6}{c}{$\gamma (5,1)\times \gamma(5,1)\times B(10,0.4)$}\\ %\hline
          \cline{2-20}
          \multicolumn{1}{c}{} &Copula & $Ind.$  &  \multicolumn{2}{c|}{$Ga$} & $Cl$ & ${Gu}$ & ${Joe}$& $Ind.$   &  \multicolumn{2}{c|}{${Ga}$} & ${Cl}$ & ${Gu}$ & ${Joe}$& $Ind.$  &  \multicolumn{2}{c|}{${Ga}$} & ${Cl}$ & ${Gu}$ & ${Joe}$ \\ %\hline
                 \cline{2-20}
          & \multicolumn{1}{|r|}{$\theta_1$}&& 0 & 0 & 0 & 0 & 0 & &0 &0 &  & 0 & 0 && 0 & 0 & 0 & 0 & 0 \\
        \multicolumn{1}{c}{Test} & \diagbox{$(a,b)$}{$\theta_2$} && 0.35 & 0.55 & 0.5 & 1.5 & 2 && 0.35 &0.55 & 0.5 & 1.5 & 2 && 0.35 & 0.55 & 0.5 & 1.5 & 2 \\
         \hline\hline
         \multicolumn{2}{c}{}& \multicolumn{16}{c}{$n=20$} \\ \hline 
           \rotatebox[]{90}{$\widetilde{D}_n$} & ~ & 5 & 29 & 58 & 15 & 60 & 70 & 5 & 24 & 52 & 14 & 56 & 67 & 5 & 20 & 48 & 16 & 44 & \cellcolor{gray!35}{\bf 55} \\ \hline
           \multirow{ 12}{*}{\rotatebox[]{90}{$\widetilde{T}_n$}} & (0.2, 0.5) & 5 & 28 & 50 & 33 & 40 & 36 & 6 & 32 & 57 & 35 & 45 & 43 & 5 & 28 & 46 & 30 & 36 & 31 \\ 
        ~ & (0.2, 1) & 5 & 27 & 47 & 30 & 36 & 34 & 6 & 31 & 56 & 32 & 45 & 42 & 5 & 27 & 45 & 30 & 36 & 31 \\
\hline
    \end{tabular}}
     \caption{Rejection rates for testing total independence of random vector with two continuous and one discrete component ($\alpha = 0.05$, $n = 20$, $n = 50$)}
    \label{powerMV21_tot}
\end{table}

\begin{table}[!h] 
    \centering 
     \resizebox{0.9\textwidth}{!}{
    \begin{tabular}{c|c|r|r|r|r|r|r|r|r|r|r|r|r|r|r|r|r|r|r}
        \multicolumn{2}{c|}{}& \multicolumn{6}{c|}{$E(1.5)^3\times P(2)^3$} & \multicolumn{6}{c|}{$E(1.5)^3\times NB(2,0.4)^3$} & \multicolumn{6}{c}{$\gamma(5,1)^3\times B(10,0.4)^3$}\\ %\hline
        \cline{2-20}
          \multicolumn{1}{c}{} &Copula & $Ind.$  &  \multicolumn{2}{c|}{$Ga$} & $Cl$ & ${Gu}$ & ${Joe}$& $Ind.$   &  \multicolumn{2}{c|}{${Ga}$} & ${Cl}$ & ${Gu}$ & ${Joe}$& $Ind.$  &  \multicolumn{2}{c|}{${Ga}$} & ${Cl}$ & ${Gu}$ & ${Joe}$ \\ %\hline
          \cline{2-20}
                  & \multicolumn{1}{|r|}{$\theta_1$}&& 0 & 0 & 0 & 0 & 0 & &0 &0 &  & 0 & 0 && 0 & 0 & 0 & 0 & 0 \\
      %   &\multicolumn{1}{|r|}{$\theta_2$}& &0.35 & 0.55 & 0.5 & 1.5 & 2 && 0.35 &0.55 & 0.5 & 1.5 & 2 &&0.35 & 0.55 & 0.5 & 1.5 & 2 \\
        \multicolumn{1}{c}{Test} & \diagbox{$(a,b)$}{$\theta_2$} && 0.35 & 0.55 & 0.5 & 1.5 & 2 && 0.35 &0.55 & 0.5 & 1.5 & 2 && 0.35 & 0.55 & 0.5 & 1.5 & 2 \\ \hline\hline
         \multicolumn{2}{c}{}& \multicolumn{16}{c}{$n=20$} \\ \hline
     \rotatebox[]{90}{$\widetilde{D}_n$} & ~ & 5 & 43 & 86 & 29 & 82 & 90 & 5 & 34 & 84 & 24 & 80 & 92 & 6 & 34 & 80 & 30 & 77 & 86 \\ \hline
       \multirow{ 12}{*}{\rotatebox[]{90}{$\widetilde{T}_n$}}  & (0.2, 0.5) & 5 & 64 & 85 & 66 & 78 & 79 & 5 & 69 & 89 & 68 & 85 & 86 & 5 & 65 & 85 & 63 & 81 & 82 \\ 
        ~ & (0.2, 1) & 5 & 64 & 85 & 65 & 78 & 79 & 5 & 70 & 90 & 67 & 86 & 87 & 5 & 66 & 86 & 65 & 82 & 83  \\
\hline
    \end{tabular}}
    \caption{Rejection rates for testing total independence of random vector with three continuous and three discrete components ($\alpha = 0.05$, $n = 20$, $n = 50$)}
    \label{powerM}
\end{table}

\section{Real data}\label{sec: real}

\subsubsection*{Example 1: Rainfall in Samaru Zaria}

We revisit the rainfall study  \cite[Table 2]{lawal2018changing}, which summarises \emph{annual} totals recorded at the Institute for Agricultural Research in Samaru Zaria, Kaduna State (Nigeria) over the period 1961–2017 (\(n = 57\) years).  
The two variables of interest are $X$—the total annual rainfall (mm)—and $Y$—the annual number of rainy days.  
While the original authors reported an inverse (negative) association between \(X\) and \(Y\), our goal is to formally test the hypothesis of independence against a general alternative of non-independence.  
Unlike \(D_n\) (p-value \(0.24\)), all our tests indicate rejection of the null hypothesis (the obtained \(p\)-values are in the range \(0.014\!-\!0.020\)).

{
\subsubsection*{Example 2: Seizure Counts for Epileptics}

The \textit{Epil Dataset} (\cite{thall1990some}) contains daily counts of epileptic seizures recorded from patients in a clinical study, alongside continuous measurements of physiological variables such as blood pressure and heart rate. 

We aim to test the independence between the count variable representing the number of epileptic seizures (base) and the continuous variable age of the patients (age). Determining whether seizure frequency is related to age is crucial for understanding the underlying factors influencing epilepsy and for tailoring appropriate treatment strategies. 

All tests reject the null hypothesis of independence (the obtained \(p\)-values for our tests are in the range \(0.012\!-\!0.017\), while for $D_n$ is less then $10^{-3}$). This indicates a significant dependence between these variables, suggesting that age plays an important role in seizure frequency. 
}

    \subsubsection*{Example 3: Bike Sharing Data}

The third real–data illustration uses the \textit{Bike Sharing Dataset}, which records daily and hourly counts of rented bicycles from the Washington, D.\,C.\ Capital Bikeshare system during 2011–2012. Alongside rental counts, the data set provides key meteorological variables (temperature, humidity, wind speed) and calendar indicators (season, working day, holiday), making it ideal for assessing how environmental factors influence demand.

Our aim is to test whether bicycle demand is independent of two weather covariates—temperature and wind speed. Let \(X_{1}\) denote temperature, \(X_{2}\) wind speed, and \(Y\) the daily rental count. We consider the null hypothesis of independence between  vector
 $(X_{1}, X_{2})$  and  $Y$,
against the alternative that dependence exists.

To minimize seasonality effects and serial autocorrelation, we draw a simple random sample of \(n = 200\) days from the full daily series and treat these observations as approximately independent. The values are given in Table \ref{tab:data3}. This subset provides a clean basis for applying the proposed independence tests.
All our tests indicate rejection of the null hypothesis, whereas the \(p\)-value for \(D_n\) is 0.62.  
This clearly demonstrates the advantage of the novel statistics.

\begin{table}[h!]
    \centering
       \resizebox{0.9\textwidth}{!}{
   \begin{tabular}{rrrr|rrrr|rrrr}
Instant & Temp & Windspeed & Count & Instant & Temp & Windspeed & Count & Instant & Temp & Windspeed & Count \\
\hline
9 & 0.138333 & 0.36195 & 822  &  245 & 0.643333 & 0.139929 & 4727  &  490 & 0.6275 & 0.162938 & 6296  \\
12 & 0.172727 & 0.304627 & 1162  &  246 & 0.669167 & 0.185325 & 4484  &  491 & 0.621667 & 0.152992 & 6883  \\
14 & 0.16087 & 0.126548 & 1421  &  248 & 0.673333 & 0.212696 & 3351  &  503 & 0.593333 & 0.229475 & 7384  \\
15 & 0.233333 & 0.157963 & 1248  &  255 & 0.644348 & 0.088913 & 4713  &  506 & 0.620833 & 0.254367 & 7129  \\
18 & 0.216667 & 0.146775 & 683  &  257 & 0.673333 & 0.1673 & 4785  &  508 & 0.615 & 0.118167 & 6073  \\
30 & 0.216522 & 0.0739826 & 1096  &  260 & 0.491667 & 0.189675 & 4511  &  511 & 0.68 & 0.14055 & 6734  \\
33 & 0.26 & 0.264308 & 1526  &  262 & 0.549167 & 0.151742 & 4539  &  516 & 0.656667 & 0.134329 & 6855  \\
34 & 0.186957 & 0.277752 & 1550  &  275 & 0.356667 & 0.222013 & 2918  &  517 & 0.68 & 0.195279 & 7338  \\
35 & 0.211304 & 0.127839 & 1708  &  278 & 0.538333 & 0.17725 & 4826  &  519 & 0.583333 & 0.186562 & 8120  \\
38 & 0.271667 & 0.0454083 & 1712  &  282 & 0.540833 & 0.06345 & 5511  &  520 & 0.6025 & 0.184087 & 7641  \\
44 & 0.316522 & 0.260883 & 1589  &  287 & 0.550833 & 0.223883 & 3644  &  523 & 0.554167 & 0.077125 & 7055  \\
49 & 0.521667 & 0.264925 & 2927  &  291 & 0.5325 & 0.110087 & 4748  &  524 & 0.6025 & 0.15735 & 7494  \\
52 & 0.303333 & 0.307846 & 1107  &  293 & 0.475833 & 0.422275 & 4195  &  528 & 0.720833 & 0.207713 & 6664  \\
53 & 0.182222 & 0.195683 & 1450  &  295 & 0.4225 & 0.0926667 & 4308  &  530 & 0.655833 & 0.343279 & 7421  \\
58 & 0.343478 & 0.125248 & 2402  &  297 & 0.463333 & 0.118792 & 4187  &  532 & 0.639167 & 0.176617 & 7665  \\
62 & 0.198333 & 0.225754 & 1685  &  302 & 0.254167 & 0.351371 & 627  &  537 & 0.7825 & 0.113812 & 6211  \\
63 & 0.261667 & 0.203346 & 1944  &  305 & 0.400833 & 0.135571 & 4068  &  547 & 0.765 & 0.161071 & 5687  \\
68 & 0.295833 & 0.22015 & 1891  &  309 & 0.326667 & 0.189062 & 3926  &  548 & 0.815833 & 0.168529 & 5531  \\
69 & 0.389091 & 0.261877 & 623  &  312 & 0.408333 & 0.0690375 & 4205  &  552 & 0.8275 & 0.194029 & 6241  \\
70 & 0.316522 & 0.23297 & 1977  &  318 & 0.53 & 0.306596 & 4486  &  559 & 0.715833 & 0.146775 & 7446  \\
72 & 0.384348 & 0.270604 & 2417  &  319 & 0.53 & 0.199633 & 4195  &  562 & 0.745833 & 0.166667 & 6031  \\
75 & 0.365217 & 0.203117 & 2192  &  322 & 0.274167 & 0.168533 & 3392  &  563 & 0.763333 & 0.164187 & 6830  \\
76 & 0.415 & 0.209579 & 2744  &  323 & 0.329167 & 0.224496 & 3663  &  565 & 0.793333 & 0.137442 & 5713  \\
77 & 0.54 & 0.231017 & 3239  &  324 & 0.463333 & 0.18595 & 3520  &  571 & 0.750833 & 0.211454 & 7592  \\
81 & 0.441667 & 0.22575 & 2703  &  325 & 0.4475 & 0.138054 & 2765  &  572 & 0.724167 & 0.1648 & 8173  \\
83 & 0.285 & 0.243787 & 1865  &  333 & 0.458333 & 0.258092 & 2914  &  574 & 0.781667 & 0.152992 & 6904  \\
88 & 0.3025 & 0.226996 & 2425  &  335 & 0.3125 & 0.220158 & 3727  &  576 & 0.721667 & 0.170396 & 6597  \\
89 & 0.3 & 0.172888 & 1536  &  336 & 0.314167 & 0.100754 & 3940  &  580 & 0.7525 & 0.129354 & 7261  \\
91 & 0.3 & 0.258708 & 2227  &  339 & 0.385833 & 0.0622083 & 3811  &  581 & 0.765833 & 0.215792 & 7175  \\
93 & 0.378333 & 0.182213 & 3249  &  341 & 0.41 & 0.266175 & 705  &  583 & 0.769167 & 0.290421 & 5464  \\
99 & 0.3425 & 0.133083 & 2455  &  343 & 0.290833 & 0.0827167 & 3620  &  587 & 0.755833 & 0.1561 & 7286  \\
100 & 0.426667 & 0.146767 & 2895  &  346 & 0.238333 & 0.06345 & 3310  &  590 & 0.700833 & 0.122512 & 6544  \\
103 & 0.4125 & 0.250617 & 2162  &  350 & 0.375 & 0.260575 & 3577  &  591 & 0.720833 & 0.136212 & 6883  \\
108 & 0.5125 & 0.163567 & 3429  &  354 & 0.385833 & 0.0615708 & 3750  &  593 & 0.706667 & 0.169771 & 7347  \\
109 & 0.505833 & 0.157971 & 3204  &  356 & 0.423333 & 0.047275 & 3068  &  595 & 0.723333 & 0.231354 & 7148  \\
111 & 0.459167 & 0.325258 & 4189  &  360 & 0.321739 & 0.239465 & 1317  &  600 & 0.6675 & 0.0702833 & 7375  \\
114 & 0.581667 & 0.192175 & 4191  &  362 & 0.29913 & 0.293961 & 2302  &  604 & 0.653333 & 0.228858 & 5255  \\
115 & 0.606667 & 0.185333 & 4073  &  367 & 0.273043 & 0.329665 & 1951  &  618 & 0.61 & 0.224496 & 8227  \\
118 & 0.6175 & 0.320908 & 4058  &  372 & 0.393333 & 0.174758 & 4521  &  619 & 0.583333 & 0.258713 & 7525  \\
120 & 0.4725 & 0.235075 & 5312  &  379 & 0.18 & 0.187183 & 2493  &  620 & 0.5775 & 0.0920542 & 7767  \\
139 & 0.530833 & 0.108213 & 4575  &  382 & 0.373043 & 0.34913 & 2935  &  631 & 0.65 & 0.283583 & 8395  \\
141 & 0.6025 & 0.12065 & 5805  &  383 & 0.303333 & 0.415429 & 3376  &  641 & 0.590833 & 0.104475 & 4639  \\
142 & 0.604167 & 0.148008 & 4660  &  386 & 0.173333 & 0.222642 & 1301  &  647 & 0.383333 & 0.189679 & 5478  \\
146 & 0.708333 & 0.199642 & 4677  &  390 & 0.294167 & 0.161071 & 4270  &  648 & 0.446667 & 0.1903 & 6392  \\
158 & 0.7075 & 0.187808 & 4833  &  391 & 0.341667 & 0.0733958 & 4075  &  649 & 0.514167 & 0.187821 & 7691  \\
160 & 0.808333 & 0.149883 & 3915  &  398 & 0.399167 & 0.187187 & 3761  &  661 & 0.4875 & 0.0814833 & 7058  \\
166 & 0.626667 & 0.167912 & 5180  &  400 & 0.264167 & 0.121896 & 2832  &  664 & 0.55 & 0.124375 & 7359  \\
174 & 0.728333 & 0.238804 & 4790  &  402 & 0.282609 & 0.1538 & 3784  &  670 & 0.3575 & 0.166667 & 5566  \\
184 & 0.716667 & 0.228858 & 4649  &  403 & 0.354167 & 0.147379 & 4375  &  671 & 0.365833 & 0.157346 & 5986  \\
186 & 0.746667 & 0.126258 & 4665  &  404 & 0.256667 & 0.133721 & 2802  &  677 & 0.295833 & 0.304108 & 5035  \\
189 & 0.709167 & 0.225129 & 4040  &  410 & 0.319167 & 0.141179 & 3922  &  682 & 0.485 & 0.173517 & 6269  \\
190 & 0.733333 & 0.167912 & 5336  &  412 & 0.316667 & 0.091425 & 3005  &  684 & 0.289167 & 0.199625 & 5495  \\
191 & 0.7475 & 0.183471 & 4881  &  413 & 0.343333 & 0.205846 & 4154  &  686 & 0.345 & 0.171025 & 5698  \\
197 & 0.686667 & 0.208342 & 5923  &  418 & 0.395833 & 0.234471 & 4773  &  697 & 0.291667 & 0.237562 & 3959  \\
198 & 0.719167 & 0.245033 & 5302  &  419 & 0.454167 & 0.190913 & 5062  &  700 & 0.298333 & 0.0584708 & 5668  \\
201 & 0.768333 & 0.113817 & 4332  &  428 & 0.414167 & 0.161079 & 4066  &  701 & 0.298333 & 0.0597042 & 5191  \\
204 & 0.849167 & 0.131221 & 3285  &  432 & 0.404167 & 0.345779 & 4916  &  702 & 0.3475 & 0.124379 & 4649  \\
207 & 0.771667 & 0.200258 & 4590  &  436 & 0.361739 & 0.222587 & 4911  &  704 & 0.475833 & 0.174129 & 6606  \\
214 & 0.783333 & 0.20585 & 4845  &  438 & 0.565 & 0.23695 & 5847  &  705 & 0.438333 & 0.324021 & 5729  \\
216 & 0.71 & 0.19715 & 4576  &  443 & 0.4725 & 0.126883 & 5892  &  706 & 0.255833 & 0.174754 & 5375  \\
219 & 0.7425 & 0.201487 & 3785  &  444 & 0.545 & 0.162317 & 6153  &  714 & 0.281667 & 0.131229 & 5611  \\
221 & 0.775 & 0.151121 & 4602  &  450 & 0.4375 & 0.220775 & 4996  &  715 & 0.324167 & 0.10635 & 5047  \\
228 & 0.700833 & 0.236329 & 4725  &  458 & 0.433913 & 0.312139 & 5936  &  716 & 0.3625 & 0.100742 & 3786  \\
231 & 0.685 & 0.139308 & 4153  &  464 & 0.5 & 0.232596 & 5169  &  720 & 0.33 & 0.132463 & 4128  \\
233 & 0.710833 & 0.248754 & 3873  &  469 & 0.4425 & 0.155471 & 6398  &  721 & 0.326667 & 0.374383 & 3623  \\
234 & 0.691667 & 0.27675 & 4758  &  478 & 0.396667 & 0.344546 & 1027  &  724 & 0.231304 & 0.0772304 & 920  \\
244 & 0.655 & 0.141796 & 5115  &  489 & 0.56 & 0.133696 & 6421  &   &  &  &   \\
\hline
\end{tabular}}
    \caption{Sample of 200 element from Bike Sharing Data}
    \label{tab:data3}
\end{table}

 \subsubsection*{Example 4: Horseshoe Crab Mating}
 
The dataset \texttt{crabs} (available  in R-package MASS) contains information on female horseshoe crabs, including physical characteristics such as weight, carapace width, color, and spine length. 
Larger, healthier females tend to attract more satellite males, suggesting dependence between female traits and satellite behavior.

 Here we  test the independece of  $(X_1,X_2)$ and $Y$, where  $X_1$ is  a carapace width, $X_2$ is a crab's weight and $Y$ is  the number of attracted males satellites.  P-values  of all statistics were less than $10^{-3}$ indicating  a strong dependence among the considered variables.

\section{Conclusion}
In this paper, we introduced two novel and powerful classes of flexible test statistics designed for assessing dependence among variables of different types—a challenge that frequently arises in practical applications. A key advantage of the bivariate version of the proposed statistics is its natural extension to multivariate settings, allowing for both the testing of independence between random vectors and the assessment of total independence among all components. In these settings, the novel statistics often outperformed existing competitors. The strong empirical performance suggests that extending these methods to high-dimensional scenarios may be a promising direction for future research.

\section*{Acknowledgments}

The work of B. Milošević and M. Cuparić is supported by the Ministry of Science, Technological Development and Innovations of the Republic of Serbia (the contract 451-03-136/2025-03/200104) and also supported by the COST action CA21163 - Text, functional and other high-dimensional data in econometrics: New models, methods, applications (HiTEc).

%\section*{Conflict of interest} 
%The authors declare that they have no conflict of interest.

\begin{appendices}

\end{appendices}

\end{document}